\documentclass{article}

\usepackage{arxiv}

\usepackage[utf8]{inputenc} 
\usepackage[T1]{fontenc}    
\usepackage[hyphens]{url}            
\usepackage{booktabs}       
\usepackage{amsfonts}       
\usepackage{amssymb}
\usepackage{amsmath, mathtools, amsthm} 
\usepackage{thmtools}
\usepackage{float}
\usepackage{nicefrac}       
\usepackage{microtype}      
\usepackage{lipsum}
\usepackage{graphicx}
\usepackage{caption}
\usepackage{subcaption}
\usepackage{multirow, array} 
\usepackage{todonotes}
\usepackage{footnote} 
\usepackage[normalem]{ulem}
\usepackage{cancel} 
\usepackage[bottom]{footmisc} 
\usepackage{arydshln} 

\usepackage{color}
\usepackage{dsfont} 

\usepackage[
backend=biber,
style=numeric,
citestyle=numeric-comp,
sorting = nyt,
maxbibnames = 99
]{biblatex}

\renewbibmacro{in:}{} 
\DeclareFieldFormat{pages}{#1} 

\usepackage{tikz,pgfplots,pgf}
\usetikzlibrary{matrix,shapes,positioning}
\usetikzlibrary{arrows,backgrounds}
\usepgflibrary{shapes.multipart}

\usepackage[many]{tcolorbox}
\newtcolorbox{cross}{blank,breakable,parbox=false,
  overlay={\draw[red,line width=2pt] (interior.south west)--(interior.north east);
    \draw[red,line width=2pt] (interior.north west)--(interior.south east);}}

\declaretheoremstyle[notefont=\normalfont]{normalhead}
\declaretheorem[style=normalhead]{theorem}
\declaretheorem[style=normalhead]{proof*}
\declaretheorem[style=normalhead]{definition}
\declaretheorem[style=normalhead]{corollary}
\declaretheorem[style=normalhead]{lemma}
\declaretheorem[style=normalhead]{proposition}

\declaretheorem[style=normalhead]{remark}

\DeclareCaptionFormat{cont}{#1 (cont.)#2#3\par}

\newcommand{\editcolor}{black} 

\newcommand\MyBox[2]{
  \fbox{\lower0.75cm
    \vbox to 1.7cm{\vfil
      \hbox to 1.7cm{\hfil\parbox{1.4cm}{#1\\#2}\hfil}
      \vfil}%
  }%
}
\title{Modeling surrender risk in life insurance: theoretical and experimental insight }

\author{
    Mark Kiermayer \\ 
  Department of Natural Science\\
  University of Applied Sciences Ruhr West\\
 \texttt{mark.kiermayer@hs-ruhrwest.de} \\ 
}

\begin{document}
\maketitle
\raggedbottom

\begin{abstract}
    Surrender poses one of the major risks to life insurance and a sound modeling of its true probability has direct implication on the risk capital demanded by the Solvency II directive.  We add to the existing literature by performing extensive experiments that present highly practical results for various modeling approaches, including XGBoost, random forest, GLM and neural networks. Further, we detect shortcomings of prevalent model assessments, which are in essence based on a confusion matrix. Our results indicate that accurate label predictions and a sound modeling of the true probability can be opposing objectives. We illustrate this with the example of resampling. While resampling is  capable of improving label prediction in rare event settings, such as surrender, and thus is commonly applied, we show theoretically and numerically that models trained on resampled data predict significantly biased event probabilities. Following a probabilistic perspective on surrender, we further propose  time-dependent confidence bands on predicted mean surrender rates as a complementary assessment and demonstrate its benefit. This evaluation takes a very practical, going concern perspective, which respects that the composition of a portfolio, as well as the nature of underlying risk drivers might change over time. 
\end{abstract}

\keywords{confidence predictions, confidence bands, resampling, ensemble techniques, neural networks, XGBoost}

\section{Introduction} \label{section_introduction}

Managing risks is at the core of insurance business. With the publication of the Solvency II directive in 2009 European insurers are required to work towards a holistic view of risk. In particular, it states individual modules for the underwriting risk, market risk and counterparty default risk. For life insurance business, the Solvency II directive requires the underwriting risk to explicitly indicate the risk of individual sub-components like mortality, longevity, morbidity, life-expense, revision, lapse and life-catastrophe. \textcolor{\editcolor}{Capital requirements, as e.g. for the mortality sub-module, are then based on adverse changes or shocks on the underlying assumptions and are determined by the tail of the respective loss variable.} In the present work we look at a sound estimation of surrender risk\textcolor{\editcolor}{, a sub-module which, despite its economic relevance, see e.g. \cite{Kling.2014, Burkhart.2018},  has historically received less attention than mortality or longevity}. \textcolor{\editcolor}{As risk within the Solvency II directive is measured by the $99.5\%$ quantile of the loss variable,} we place a particular focus on \textcolor{\editcolor}{unbiased modeling of the distribution of surrender of policyholders and how}   resampling schemes, that are commonly used to mitigate the rare event character of surrender, \textcolor{\editcolor}{affect this distribution}. \\

Lapse risk has been identified as a major risk to life insurance business by the Quantitative Impact Study QIS5 of the European Insurance and Occupational Pensions Authority  and related research, see e.g. \cite{Kling.2014, Burkhart.2018}. In \cite{EuropeanComission.2010}, lapse risk is officially defined as "all legal or contractual policyholder options which can significantly change the value of the future cash-flows". Hence, lapse risk includes not only a premature and full termination of contracts, but for example also partial terminations, changes to the frequency or quantity of premium payments, altered benefits or extension of coverage. For accuracy, we use the term \textit{surrender} to specifically refer to a premature, full termination of a contract induced by the policyholder. Colloquially, however, surrender and lapse are often used interchangeably. Overall, the effects of surrender are manifold. Arguably most harmful to the insurer are early surrenders that cost the insurer its initial expenses or adverse selection that alters the composition of the portfolio and compromises a sound diversification. The review in \cite{Eling.2013} provides an extensive overview of past research on surrender risk and divides the respective literature into two subsets. Given a theoretical framework, we can determine a fair value estimate of the embedded options, see e.g. \cite{Bauer.2006, Kling.2014, Kolkiewicz.2006}. This theoretical perspective is not limited to risk-neutral individuals, but can be extended to risk-avers and non-rational consumer behavior, see \cite{Eling.2013}. On the other hand, there exists extensive empirical research, which aims to identify the main risk drivers for lapse events and to test for modeling hypotheses, as well as the quality of modeling approaches to capture the underlying risk, see e.g. \cite{Aleandri.2017, Burkhart.2018, Milhaud.2011, Yu.2019, Eling.2014, Kiesenbauer.2012}. The most common modeling classes include logistic regressions and tree based methods as CART and random forests, where the evaluation is typically performed with metrics on the respective confusion matrix or the related receiver operating curve (ROC) for assessment and model choice. The work in \cite{Aleandri.2017, Burkhart.2018, Loisel.2011, Loisel.2019} indicate how the quality of empirical lapse models translates to economic measures. In \cite{Barsotti.2016} we see the inclusion of contagion effects in the modeling of lapse risk. We note that the majority of empirical modeling approaches focus on a single-period lapse prediction, which is to a large extent motivated by the one-year $\text{VaR}_{0.995}$ risk measure imposed by the Solvency II directive. A multi-period perspective on lapse dynamics including a proposition for lapse tables, that can be used in similar fashion to standard mortality tables, can be found in \cite{Milhaud.2018}. The authors of \cite{Milhaud.2018} apply survival analytical methods, as e.g. the Fine\&Gray model, see \cite{Fine.1999}, in order to distinguish between competing risks that can lead to the termination of a contract. \\

In the present work, \textcolor{\editcolor}{we take a semi-empirical, semi-theoretical perspective. On simulated data, we follow a classical empirical approach, where we investigate the quality of common modeling approaches for surrender risk, such as logistic regression, tree based methods and neural networks. At the same time, we grant ourselves insight into the true, underlying surrender probabilities and monitor the generalization of models and how well latent information can be extracted from observable quantities. For resampling techniques commonly used in practice, we also investigate and explain the observed model behaviour from a theoretical point of view. By that we hope to provide valuable insight for all practitioners.} \textcolor{\editcolor}{Our experiments include} four different portfolios of endowment policies, each showing surrender behaviours that replicate findings reported in the literature, see \cite{Milhaud.2011, Eling.2014, Cerchiara.2008}. Overall, we find highly competitive results, where XGBoost consistently displays superior performance. For model evaluation, we look at the latent, true surrender probabilities, as well as observable realizations of surrender events. In view of the bias of a model, our results indicate a conflict between accurate label predictions and sound confidence predictions. Although, we can confirm common resampling techniques to improve the $F_1$-score, we theoretically and numerically show that resampling comes at the cost of a significant bias of estimated surrender probabilities. Following this probabilistic perspective, we further introduce time-dependent confidence bands for the predicted mean surrender rate as an alternative evaluation, which indicate the uncertainty of predicted surrender probabilities. This allows us to assess the quality of a model from a very practical going concern perspective, where the composition of the underlying portfolio and the predominant risk drivers might change over time.  In particular, confidence bands highlight that adding a risk buffer to a naive baseline does not cover the surrender risk sufficiently. \textcolor{\editcolor}{At the same time, we can use confidence bands as a diagnostic tool to uncover changes in the surrender behaviour of policyholders. } \\ 

The remainder of this paper is organized as follows. We start by clarifying the objective in Section \ref{section:objective}. In Section \ref{section:model_evaluation}, we \textcolor{\editcolor}{ recap} common, frequentist concepts to evaluate binary classifiers and \textcolor{\editcolor}{discuss } probabilistic alternatives. Next, in Section \ref{section:rare_events} we review the rare event setting of surrender risk and analyze the distributional effect of common resampling methods. Section \ref{section:data} presents the simulation and preparation of all data used in our numerical experiments. In Sections \ref{section:models} and \ref{section:num_experiments}, we then describe the specific parameterization of all classifiers and report numerical results. Lastly, we conclude in Section \ref{section:conclusion} \textcolor{\editcolor}{and provide} an outlook for future work. 

\section{Objective} \label{section:objective}
Let the probability space $(\Omega,\mathcal{F}, \mathbb{P})$ describe our stochastic environment. We denote an arbitrary insurance contract at time $t\in\mathbb{N}_0$ by $X_t: \Omega\rightarrow \mathbb{R}^n$, $n\in\mathbb{N}$. The time series $(X_0, X_1, X_2, \ldots, X_T)$ then represents the evolution of the contract up to the random period of termination $T:\Omega\rightarrow\mathbb{N}_{\geq 0}$, recording e.g. the age of the policyholder or the assured benefits.
As the termination can have multiple causes, like surrender, death or maturity, we record the respective state at the end of period $t$ by $J_{t}: \Omega\rightarrow\mathbb{N}_0$, where $J_t=0$ for $t=0,\ldots, T -1$ indicates an active contract. All states $J_{ T}\in\mathbb{N}$ are terminal, absorbing states. They describe competing, censoring events, in the sense that observing $\lbrace J_{ T}=i\rbrace$ and $\lbrace J_{ T }=j\rbrace$, $i\neq j$, are mutually exclusive, both end the observation on the time series  and prevent a realization of an alternative state. For more detail on censoring and competing risks we refer the reader to \cite{Aalen.2008}.

\paragraph{Single period setting.}In this work we consider three terminal states, namely surrender, death and maturity. Surrender is represented by $\lbrace J_{ T }=1\rbrace$. Death and maturity are indicated by $\lbrace J_{ T }=2\rbrace$ and $\lbrace J_{ T }=3\rbrace$, but will not be modeled explicitly. Further, in line with the Solvency II directive, we focus on a single period setting. Hence, we define targets $Y_t$ by $Y_t:=\mathds{1}_{\lbrace J_t=1 \rbrace}$. Given a realization $x_t$ of $X_t$, our objective is to estimate the probability of the active contract $x_t$ to surrender within the period $[t,t+1)$, i.e. to obtain an estimate for the conditional probability 
\begin{align} \label{eq:objective_prob} 
    &p({1\vert x_t}):=\mathbb{P}(Y_{t}=1\vert X_t=x_t),
    \intertext{or, equivalently, to model the \textcolor{\editcolor}{Bernoulli} random variable}
    &Y_{t}\vert (X_t=x_t)\sim\text{Ber}\left(p(1\vert x_t)\right).
\end{align}
It is important to note that \eqref{eq:objective_prob} implicitly imposes the Markov property that the random state $J_t$ and the past represented by $X_{t-k}=x_{t-k}$ are independent for $k=1,\ldots,t$. In the literature, it is common to assume that the time horizon influencing surrender decisions of policyholders is restricted to one period, see e.g. \cite{Aleandri.2017, Milhaud.2011}. Also, conditioning on an active contract $x_t$ concurrently implies $T\geq t$. \\

Overall, our objective is a common setup for supervised learning, where we aim to obtain a model $\hat{p}: \mathbb{R}^n\rightarrow [0,1],~x\mapsto \hat{p}(1\vert x)$, such that it minimizes the loss function $\ell: [0,1]\times \lbrace 0,1 \rbrace\rightarrow\mathbb{R}$, i.e.
\begin{align}
    \hat{p} := \arg\!\min_{\Bar{p}}\ell (\Bar{p}(1\vert X_t),Y_t).
\end{align}

We would like to draw attention to the non-binary output of the model $\hat{p}$. The contrast between binary and non-binary prediction will be a major topic throughout this paper. In the style of \cite{RobertE.Schapire.1999}, we call estimates $\hat{p}({1\vert x_t})\in[0,1]$ \textit{confidence prediction}. In contrast, we refer to binary predictions $\hat{y}_t\in\lbrace 0,1 \rbrace$, that purely focus on the realization of $Y_{t}\vert (X_t=x_t)$, as \textit{label predictions}. Naturally, any confidence prediction $\hat{p}({1\vert x_t})$ can be transformed to a label prediction based on some threshold $c\in[0,1]$, i.e. $\hat{y}_t=\mathds{1}_{\lbrace \hat{p}({1\vert x_t}) \geq c \rbrace}$. \\
In practice, obtaining a sound model $\hat{p}$ is complicated by the rare event character of surrender. We discuss rare events separately in Section \ref{section:rare_events}, with a specific focus on their effect on confidence predictions and the distributional effect of common resampling strategies. First, we would like to conclude this section by commenting on the setting at hand and review approaches to evaluate models in Section \ref{section:model_evaluation}.

\paragraph{Context for alternative settings.} The presented objective is customized to the Solvency II directive and a one-year perspective, which disregards the exact time of surrender within the year. Other settings, e.g. a risk neutral product design, might ask for a more general setup where the risk of individual events $\lbrace J_T=i\rbrace$, $i\in\mathbb{N}_0$, are estimated jointly in a multi-period setting. In particular, this requires a more general, real-valued random event $T:\Omega\rightarrow \mathbb{R}_{\geq 0}$ and random state variable $J_t: \Omega\rightarrow\mathbb{N}_0$, $0\leq t \leq T$. Given a realization $x_t$ of $X_t$, the objective then changes to estimating the \textit{cause-specific hazard rate} $\alpha_{T,i}(\Tilde{t}\vert x_t)$ of the terminal events $i\in J_t(\Omega)$ at time $\Tilde{t}\geq t$. According to \cite{Aalen.2008}, the quantity $\alpha_{T,i}(\Tilde{t}\vert x_t)$ is defined by
\begin{align} \label{eq:cs_hazard}
    \alpha_{T,i}(\Tilde{t}\vert x_t):=\lim_{\Delta\rightarrow 0}\tfrac{1}{\Delta}\mathbb{P}\left(\Tilde{t}\leq T \leq \Tilde{t}+\Delta,~J_T=i\vert~ T\geq \Tilde{t},~x_t\right).
\end{align}
Note that conditioning on an active contract $x_t$ also implies that $J_{\Tilde{t}}=0$ for all $\Tilde{t}<t$, as in our context all states $J_{\Tilde{t}}\neq 0$ are terminal and competing risks. For a detailed discussion of $\alpha_{T,i}(\Tilde{t}\vert x_t)$ and its estimation we refer the reader to \cite{Aalen.2008, Milhaud.2018}. Instead we want to comment on the practical implications of \eqref{eq:objective_prob} and \eqref{eq:cs_hazard}. \\
Quantities in \eqref{eq:cs_hazard} can be thought of as time-dependent transition intensities from the active state $0$ to a terminal state $i$ in a Markov model. Cause-specific hazard rates provide a more detailed insight into the likelihood of states $J_T$ to be realized and might  be used to form a table for lapse probabilities, see e.g. \cite{Milhaud.2018}. It is well know that simply ignoring competing risks, as e.g. death or maturity for surrender, alters the estimated hazards $\alpha_{T,i}(\Tilde{t}\vert x_t)$, which are asymptotically unbiased, see e.g. Section 3.4 in \cite{Aalen.2008}. 
On the other hand, modeling a Bernoulli variable as in \eqref{eq:objective_prob} translates to asking how likely is it to observe a surrender event of contract $x_t$ within the next period. Here, all non-surrender states $J_t\neq 1$ are combined by introducing the target $Y_t$. Further, in our one-period perspective we have complete observations as either surrender or non-surrender is realized. In contrast, if we do not observe contract $x$ entering a terminal state $J_T>0$, the observation on $x$ yields a right-censoring event for the estimation of \eqref{eq:cs_hazard}. \\

For notational convenience, in the remainder of this work we drop the time index $t$ in $(X_t,Y_t)$ and denote realizations $(x,y)\sim(X,Y)$ by the use of small letters. The progression of time is implicitly recorded by features like the age of the policyholder. 
Consequently, we denote the true surrender probability of realization $x$ by $p({1\vert x})$, the confidence prediction by $\hat{p}({1\vert x})$ and the label prediction by $\hat{y}$. In the following, we use indices to count data in our data set, e.g. $x_k$, and exponents, if we explicitly want to refer to components, e.g. $x_k^{(j)}$.
\section{Model evaluation} \label{section:model_evaluation}
In this section, we review classical approaches for evaluating classifiers, \textcolor{\editcolor}{indicate} shortcomings of their frequentist focus and provide probabilistic alternatives. Our reasoning is dictated by the \textcolor{\editcolor}{probabilistic} objective of sound confidence predictions, in contrast to accurate binary label predictions to model actions or definite class memberships.

\paragraph{Frequentist evaluation.} In a binary classification task, the label prediction $\hat{y}$ for a record $x$ with label $y$ can either be true, i.e. $\hat{y}=y$, or false, i.e. $\hat{y}\neq y$. Based on whether the prediction $\hat{y}$ is positive ($\hat{y}=1$) or negative ($\hat{y}=0$), we refer to a prediction as true, respectively false, positive or negative. A common \textcolor{\editcolor}{frequentist} evaluation is to present the count of all predictions $\hat{y}$ w.r.t. their correct label $y$ in a $2\times 2$ matrix, a so-called \textit{confusion matrix}, see \cite{Aleandri.2017, Milhaud.2011, Fawcett.2006}. Based on the absolute frequencies in a confusion matrix, we can then formulate a variety of performance measures. Common examples include, accuracy, precision, recall, specificity or $F_{\beta}$-score for $\beta>0$. Table \ref{table:classification_measures} in the Appendix \ref{appendix:tables} summarizes the definition of well-known performance measures. \textcolor{\editcolor}{More advanced approaches like the \textit{receiver operating characteristics} (ROC) curve, the \textit{precision-recall curve} and their quantification via the \textit{area under curve} (AUC) evaluate the model and its label predictions $\hat{y} = \mathds{1}_{\lbrace \hat{p}({1 \vert x})>c \rbrace}$ for multiple thresholds $c\in[0,1]$, see \cite{Provost.2001}. The concept of the \textit{convex hull} might be used to determine an optimal threshold or compare different classifiers, see \cite{Fawcett.2006}. Overall however, the listed frequentist evaluations show several limitations. Examples include the disregard for a full column or row of the confusion matrix, as for recall or specificity, an invariance towards changing the count of true-negative values, as for the $F_{\beta}$-score, or an insensitivity to the balance of the data, as for the ROC curve. The work in \cite{Sokolova.2009} provides a detailed summary of individual performance measures and how they are affected by permutations in the confusion matrix. More generally, one can argue that compressing a confidence prediction $\hat{p}({1 \vert x})$ to a label prediction $\hat{y}$ looses information. Hence, we will investigate an alternative, probabilistic model evaluation. }\\
    
    \paragraph{Probabilistic evaluation.} Ideally, we would like to compare the confidence predictions $\hat{p}(1\vert x)$ with true event probabilities $p({1\vert x})$, e.g. by employing $p$-$\hat{p}$ \textcolor{\editcolor}{scatter}plots. Given data $\lbrace(x_i,y_i)\rbrace_{i=1}^{N}$, we denote the mean absolute difference of these quantities by
    \begin{align}
        \text{mae}\left(p,\hat{p}\right)&:= \frac{1}{N}\sum_i~ \vert p({1\vert x_i}) - \hat{p}({1\vert x_i})\vert.
    \end{align}
    In our experiments we can access the true surrender probabilities $p(1\vert x)$ and look at the model fit from a probabilistic point of view. However, this is a purely theoretical perspective as in practice we do not observe $p(1\vert x)$ but only the binary realization $y$ from $\text{Ber}(p(1\vert x))$. Thus, information theory provides a more practical approach. One can show that the maximum likelihood estimator $\hat{p}(1\vert x)$ corresponds to minimizing the Kullback-Leibler divergence, respectively equivalently minimizing the cross-entropy, see \cite{Goodfellow.2016}. Using standard notation, we denote the empirical estimate of the (binary) cross-entropy by
    \begin{align}
        H(\hat{p}) &:= -\frac{1}{N}\sum_{i}y_i\log \hat{p}({1\vert x_i}) + (1-y_i)\log \left(1-\hat{p}({1\vert x_i})\right).
    \end{align}
    Hence, minimizing the binary cross-entropy provides a maximum-likelihood consistent estimator, which is why the cross-entropy is the most commonly chosen loss function in classification. Unfortunately however, the cross-entropy lacks a comprehensive, practical interpretation for surrender. Therefore, we introduce a third option to evaluate confidence predictions, \textcolor{\editcolor}{i.e. confidence intervals} based on the central limit theorem of Lindeberg-Feller, see e.g. \cite{Klenke.2014}.
    
    \begin{proposition} \label{prop:CLT}
        Let $x_1,\ldots,x_N$ be realizations of contract $X$. Let $Z_1,\ldots,Z_N$ describe random, independent surrender events, where $Z_i \sim \text{Ber}\left(p({1\vert x_i})\right)$, $i=1,\ldots,N$. If there exists $\varepsilon>0$ such that $p(1\vert x)\in[\varepsilon,1-\varepsilon]$ for any contract $x\sim X$, then
        \begin{align}
            &\frac{1}{\sigma(S_N)}\sum_{i} Z_i - p({1\vert x_i}) \overset{d}{\rightarrow} \mathcal{N}\left(0,1\right), \quad N\rightarrow\infty,
            \intertext{with}
            & \sigma(S_N) := \sqrt{\text{Var}\left(\sum_i Z_i\right)} = \sqrt{\sum_i p(1\vert x_i)(1-p({1 \vert x_i}))}
        \end{align}
        holds.
    \end{proposition}
    \begin{proof}\phantom{\qedhere}
        See Appendix \ref{proof:prop:CLT}.
    \end{proof}
    
    \textcolor{\editcolor}{We note that the surrender events $Z_i$ are in general non-identically distributed. Further,} assuming $Z_i:=(Y\vert X=x_i)$ in Proposition \ref{prop:CLT} to be independent means that, given circumstances $x_i$, policyholders act independently of each other. If $x_1$ and $x_2$ present observations on the same contract, i.e. the same policyholder, at different times this assumption is implausible. Hence, in application we restrict $x_i$ in Proposition \ref{prop:CLT} to \textcolor{\editcolor}{observations within} a fixed calendar year. 
    Features that can effect the surrender risk of all policyholders holistically, as e.g. rising interest rate, can be recorded as a component of contracts $x_i$. Thus, a holistic increase of the surrender risk does not violate the independence of surrender decisions $Z_i$. Note that the technical restriction of $p(1\vert x_i)\in [\varepsilon,1-\varepsilon]$ does not impact the asymptotic distribution in Proposition \ref{prop:CLT}. We can think of it as assuming a minimum risk for surrender and non-surrender. For any finite data set $\lbrace(x_i,y_i \rbrace_{i=1}^N$ with estimator $\hat{p}(1\vert x_i)\in(0,1)$ there exists $\varepsilon>0$ which satisfies $\hat{p}(1\vert x_i)\in[\varepsilon,1-\varepsilon]$.
    
    \begin{corollary} \label{cor:confidence_intervals}
        Let the assumptions of Proposition \ref{prop:CLT} hold and let $\hat{p}: x\mapsto\hat{p}({1\vert x})$ denote an estimator of $p(1\vert x)$. Given the confidence level $\alpha\in(0,1)$, we can then construct a confidence interval for the mean surrender rate by
            \begin{align}
                \left(\frac{1}{N}\sum_i \hat{p}({1\vert x_i})\right) \pm z_{1-\alpha/2}\frac{\sqrt{\sum_i \hat{p}({1\vert x_i})(1-\hat{p}({1\vert x_i}))}}{\sqrt{N(N-1)}},
            \end{align}
        where $\hat{p}({1\vert x_i})$ presents the model estimate for $p({1\vert x_i})$ and $z_{1-\alpha/2}$ the respective percentile of the standard normal distribution. The best point estimate for the mean surrender rate is given by $\frac{1}{N}\sum_i \hat{p}({1\vert x_i})$.
    \end{corollary}
    \begin{proof}\phantom{\qedhere}
        See Appendix \ref{proof:cor:CLT}.
    \end{proof}
    If we were to work with label predictions $\hat{y}_i := \mathds{1}_{\lbrace \hat{p}(1\vert x_i)\geq 0.5\rbrace}\in\lbrace 0,1 \rbrace$, the quantity $\tfrac{1}{N}\sum_i \hat{y}_i$ also provides an estimator for $\mathbb{E}\left(\tfrac{1}{N}\sum_i Z_i\right)$. However, without an estimate of the distribution of $Z_i$ we cannot express any level of confidence for our predictions. \textcolor{\editcolor}{Note that Corollary \ref{cor:confidence_intervals} does not guarantee the confidence intervals to strictly lie within $[0,1]$. This stems from the Gaussian character of the asymptotic confidence intervals. However, we will later find this to not be an issue in practical application. } In Section \ref{section:num_experiments}, we apply Corollary \ref{cor:confidence_intervals} for each calendar year separately and obtain confidence bands on our time series. \textcolor{\editcolor}{Further, Corollary \ref{cor:confidence_intervals}} relaxes the assumptions of \cite{Milhaud.2011}, where the authors construct confidence intervals under the assumption of homogeneous and identically distributed subgroups of contracts.
\section{Rare events} \label{section:rare_events}
We now combine the nature of rare events with the objective of Section \ref{section:objective} and \textcolor{\editcolor}{check whether selecting models based on} frequentist performance measures \textcolor{\editcolor}{is aligned with our} probabilistic objective. We start with a definition of rare events, motivate resampling and examine it from a purely theoretical point of view. Thereafter, we interpret these results for common classifiers in practice. \\

In binary classification, we generally refer to data $\lbrace (x_i,y_i)\rbrace_{i=1}^{N}$ drawn from $(X,Y)$ as \textit{balanced data} if both classes contain about equally many observations. Concerns about biased collection of data are discussed e.g. in \cite{King.2001}. Overall, it has been widely acknowledged in literature, see \cite{Wallace.11.12.201114.12.2011, Maalouf.2014, Seiffert.2007, Weiss.2004}, that imbalanced data can lead to classifiers which are biased towards predicting the majority class, and that standard metrics, as e.g. accuracy, \textcolor{\editcolor}{may provide little information about the actual quality of the model}. Here, the notion of a bias refers solely to label predictions. In the particular case of the logistic regression, imbalanced data even leads to low values in the Hessian matrix \textcolor{\editcolor}{of the loss function}, see \cite{Hastie.2017}. Hence, estimated regression parameters are highly unstable, despite being asymptotically unbiased. \\
In contrast to imbalanced data,  we characterize \textit{rare events} by their low, true probability $\mathbb{P}(Y=1)<\varepsilon$, where the value of $\varepsilon$ can be domain specific, see \cite{Rubino.2009}. In practice, values as low as $\varepsilon =0.05$ or $\varepsilon =0.01$ are commonly found in e.g. surrender of insurance contracts \cite{Milhaud.2011, Milhaud.2018, Aleandri.2017} or fraud detection \cite{Phua.2012, Li.2012}. In this work, we will assume perfect data in the sense that any imbalance between surrender and non-surrender events stems from the rare event character of the task at hand and is not the consequence of improper data collection. \\

To address the concern of biased classifiers on imbalanced data, the literature proposes two main techniques, cost-sensitive-learning and resampling, see e.g. \cite{Wallace.11.12.201114.12.2011, King.2001,Maalouf.2014, Elkan.2001, Weiss.2004, Seiffert.2007}. Both have empirically shown to be capable of improving classifiers regarding their recall, $F_{\beta}$, AUC or geometric mean of the true positive and true negative rate, see e.g. \cite{Wallace.11.12.201114.12.2011, Chawla.2002, Seiffert.2007, Weiss.2001}. In \cite{RobertE.Schapire.1999}, we see that popular boosting schemes like 'AdaBoost' can be also interpreted as an adaptive resampling. For completeness we note that, research in data mining indicates that classifiers can also benefit from removing outliers or accounting for noisy measurements close to the decision boundary. For a more comprehensive review we refer the reader to \cite{Chandola.2009, Hodge.2004}. In the following, we focus on common resampling schemes, as resampling data and applying weights to the empirical loss function are conceptually equivalent. This argument is supported by the authors in \cite{Wallace.11.12.201114.12.2011, An.2021}, who \textcolor{\editcolor}{even} contend that sampling approaches will generally outperform cost sensitive learning. \\ 

We start with a purely theoretical perspective on resampling. Therefore, we denote the marginal, respectively joint, density functions of $(X,Y)$ at $(x,y)$ by $p(x)$ and $p(y)$, respectively $p({x,y})$. The specific letter indicates the underlying random vector as it is common in Bayesian literature, see e.g. \cite{Wakefield.2013}. Further, we introduce the functions
    \begin{align} \label{eq:reg_cond_dist}
        p({y\vert x}) & := 
        \begin{cases}
            {p({x,y})}\bigm/{ p({x}) }, & p(x)>0 \\
            0, & p(x) =0,
        \end{cases} 
        & p({x\vert y}) & := 
        \begin{cases}
            {p({x,y})}\bigm/{ p({y}) }, & p(y)>0 \\
            0, & p(y) =0.
        \end{cases}
    \end{align}
Based on the concept of regular conditional distributions, see \cite{Durrett.2010, Klenke.2014}, both functions in \eqref{eq:reg_cond_dist} are well defined on the given probability space $(\Omega,\mathcal{F}, \mathbb{P})$ as a function w.r.t. either $x$ or $y$. Hence, we are justified to interpret $p({y\vert x})$ as a conditional density given the condition $x$ or as a measurable function w.r.t. its condition $x$. \\

    \textcolor{\editcolor}{In this context, resampling} aims to mitigate the imbalance between classes by either dropping samples of the majority class or sampling samples from the minority class. We call this procedure \textit{random over-} respectively \textit{undersampling} if the dropping respectively sampling is performed randomly. In practice, these resampling procedures are usually performed until the data is balanced. Formally, we define resampling with the purpose of increasing the share of the minority class 1, without specifying the specific target ratio of minority and minority class, as follows.
    
    \begin{definition} \label{def:resampling}
        Let $\mathcal{D}=\lbrace(x_i,y_i)\rbrace_{i=1}^N$ be i.i.d. samples generated from $(X,Y)$, where $Y\vert(X=x) \sim \text{Ber}\left(p(1\vert x)\right)$. 
        Let the data be imbalanced with minority class $Y=1$, such that $\sum_i y_i< \sum_i (1-y_i)$. 
        We then call $S$ a \textit{resampling scheme}, 
        if it \textcolor{\editcolor}{enlarges or reduces $\mathcal{D}$ to} \textcolor{\editcolor}{data} $\mathcal{D}^S=\lbrace(x^{S}_i,y^{S}_i)\rbrace_{i=1}^{N^{S}} \subset\mathcal{D}$ with $N^S>N$ or $N^S<N$, respectively, for which
            \begin{align*}
                \sum_{0\leq i\leq N} (1-y_i)  = \sum_{0\leq i\leq N^S} (1-y^{S}_i) \quad \text{ and }  \quad \lbrace (x_i,y_i)\in\mathcal{D}:y_i=0 \rbrace \subset \mathcal{D}^S, \tag{\textit{oversampling}} \\
                \intertext{respectively}
                \sum_{0\leq i\leq N} y_i  = \sum_{0\leq i\leq N^S} y^{S}_i \quad \text{ and } \quad \lbrace (x_i,y_i)\in\mathcal{D}:y_i=1 \rbrace \subset \mathcal{D}^S \quad  \tag{\textit{undersampling}}
            \end{align*}
        hold.
    \end{definition}    
    To indicate that we refer to empirical data, we write $\hat{p}{(\cdot)},~\hat{p}{(\cdot,\cdot)}$ and $\hat{p}{(\cdot\vert\cdot)}$ for the densities induced by the original data $\lbrace(x_i,y_i)\rbrace_{i=1}^N$. \textcolor{\editcolor}{Analogously}, $\hat{p}^S{(\cdot)},~\hat{p}^S{(\cdot,\cdot)}$ and $\hat{p}^S{(\cdot\vert\cdot)}$ are derived from resampled data $\lbrace(x^{S}_i,y^{S}_i)\rbrace_{i=1}^{N^{S}}$. Next, we impose that any proper resampling scheme should not alter the distribution of features in a given class, which is natural as any estimator $\hat{p}: x\mapsto \hat{p}(1\vert x)$ relies on features $x$. 
    
    \begin{definition}[Consistent resampling] \label{def:consistent_resampling}
        We call a \textit{resampling scheme} $S$ \textit{consistent} if it preserves the observed probabilities for features $X$ within every class $Y=y$, i.e. with $y\in\lbrace 0,1 \rbrace$ fixed
            $$\hat{p}^S(x\vert y) = \hat{p}(x\vert y)$$
        holds for arbitrary $x\in X(\Omega)$.
    \end{definition} 
    
    \begin{remark}[Examples for consistent resampling] \label{remark:constistent:resampling}
        Random undersampling and oversampling, as well as the popular resampling algorithm SMOTE, see \cite{Chawla.2002}, are all natural approaches to building a consistent resampling method. All approaches alter data $x$ only for either the minority or the majority class. We provide arguments for the consistency of these resampling techniques, focusing on the class altered by the specific algorithm.
        \begin{itemize}
            \item[1.] In random undersampling, we iteratively and randomly drop observation $(x,0)\in\lbrace(x_i,y_i)\rbrace_{i=1}^N$. In the first iteration, the probability to drop $(x,0)$ is characterized by the share of observations $x_i=x$ in class $y=0$, which poses an estimator $\hat{p}(x\vert y=0)$. Assuming a sufficiently large sample size for the majority class, i.e. $y=0$, we can approximate the probability of subsequent observations $(x^{\prime},0)$ to be dropped by $\hat{p}(x^{\prime}\vert y=0)$.
            \item[2.] In random oversampling, we iteratively draw observations $(x,1)\in\lbrace(x_i,y_i)\rbrace_{i=1}^N$ with replacement. Hence, each $x$ is drawn with probability $\hat{p}(x\vert y=1)$.
            \item[3.] SMOTE oversampling, see \cite{Chawla.2002}, also draws observations $(x,1)\in\lbrace(x_i,y_i)\rbrace_{i=1}^N$ with replacement and with probability $\hat{p}(x\vert y=1)$. In contrast to random oversampling, it does not sample $x$ but $x^{\prime}= x+u\cdot \tilde{x}$, where $u\sim U(0,1)$ and $\tilde{x}$ with $(\tilde{x},1)\in \lbrace(x_i,y_i)\rbrace_{i=1}^N$ a randomly chosen, $k$-nearest-neighbour to $x$. By default, SMOTE employs $k=5$. This approach samples new data points by assuming $\hat{p}(x\vert y=1)$ to be locally constant. The notion of locality is dictated by the distance metric used in the $k$-nearest neighbour algorithm.
        \end{itemize}
    \end{remark}
    
    According to the work in \cite{Elkan.2001}, the effect of consistent resampling can be formalized as follows.
    \begin{theorem}[Probabilistic view on resampling, see \cite{Elkan.2001}] \label{theorem:effect_of_resampling}
        Let $S$ be a consistent resampling scheme and $(x,y)$ arbitrary data which satisfy $\hat{p}(x)>0$ and $\hat{p}^S(x)>0$. Then the equation
        \begin{align*}
            \hat{p}^S(1\vert x) = \hat{p}(1\vert x)~\frac{\hat{p}^S(y=1)~(1-\hat{p}(y=1))}{\hat{p}(y=1)~(1-\hat{p}(1\vert x))+\hat{p}^S(y=1)~(\hat{p}(1\vert x)-\hat{p}(y=1))}
        \end{align*}
        holds.
    \end{theorem}
    
    In Theorem \ref{theorem:effect_of_resampling}, the assumption of positives densities $\hat{p}(x)>0$ and $\hat{p}^S(x)>0$ is crucial to apply Bayes' theorem in the proof of Theorem \ref{theorem:effect_of_resampling} for both $\hat{p}$ and $\hat{p}^S$. In the simplest case, where estimators $\hat{p}$ and $\hat{p}^S$ are based on relative frequencies this certainly holds true for data $(x,y)$ that are in the original and the resampled dataset. Next, we discuss the assumption of Theorem \ref{theorem:effect_of_resampling} for common classifiers. \\
    Generalized linear models, tree based classifiers or neural networks all can be used to construct confidence predictions $\hat{p}(1\vert x)$ which are, at the very least, locally continuous. For example, a logistic regression or a neural networks with a final sigmoid activation both result in $\mathbb{P}$-almost-surely differentiable $\hat{p}(1\vert x)\in\left(0,1\right)$ for all contracts $x$ and, thus, $p(x)>0$, see \eqref{eq:reg_cond_dist}. In its standard form, see \cite{Hastie.2017}, classification trees apply relative frequencies per node, i.e. a locally constant estimator $\hat{p}(1\vert x)$. In pure nodes, where $\hat{p}(1\vert x)\in\lbrace 0,1 \rbrace$, we find $p(x,y)=0$ and $p(x)>0$, see again \eqref{eq:reg_cond_dist}. Hence, the stated models all assume $p(x)>0$ globally and show the generality of Theorem \ref{theorem:effect_of_resampling} in practice. \\
    
    Observe that for any resampling schemes as in Definition \ref{def:resampling} holds $\frac{1}{N^S}\sum_{y_i^S}y_i^S>\frac{1}{N}\sum_{y_i}y_i$. Hence, it is natural to assume that $\hat{p}^S(y=1)>\hat{p}(y=1)$ holds for modeling approaches in practice. 
    Then, we can establish the effect of resampling from a distributional perspective.
    \begin{lemma}[Bias of traditional resampling] \label{lemma:bias_resampling}
        Let $S$ be a consistent resampling scheme. Further, consider a modeling approach with $\hat{p}^S(y=1)>\hat{p}(y=1)$ and arbitrary data $(x,y)$ with $\hat{p}(x)>0$ and $\hat{p}^S(x)>0$. Then, $S$ increases the conditional distribution of the minority class 1, such that $\hat{p}^S(1\vert x)>\hat{p}(1\vert x)$ holds.
    \end{lemma}
    \begin{proof}\phantom{\qedhere}
        See Appendix \ref{proof:lemma:res}.
    \end{proof}
    
    Following the standard rules of differentiation we can further analyse the bias indicated by Lemma \ref{lemma:bias_resampling}.
    \begin{lemma} \label{lemma:skew_resampling}
        Let $S$ be a consistent resampling scheme and $(x,y)$ arbitrary data with $\hat{p}(x)>0$ and $\hat{p}^S(x)>0$. Then
        \begin{align*}
            \frac{\partial\hat{p}^S(1\vert x)}{~\partial \hat{p}(1\vert x)}>0 ~~~\text{ and } ~~~  
            \frac{\partial^2\hat{p}^S(1\vert x)}{~\partial \hat{p}(1\vert x)^2}<0
        \end{align*}
        holds.
    \end{lemma}
    
    Lemma \ref{lemma:bias_resampling} immediately implies that if $\hat{p}(1\vert x)$ is a consistent estimator, i.e. it $\mathbb{P}$-almost-surely converges to the true, latent probability $p(1\vert x)$ as the sample size $N\rightarrow \infty$, then $\hat{p}^S(1\vert x)$ formed on resampled data will not be consistent.
    In particular, if we look at Theorem \ref{theorem:effect_of_resampling} as a function of $\hat{p}(1\vert x)$ we can explicitly retrieve the relative bias introduced by resampling. As recorded in Lemma \ref{lemma:skew_resampling}, resampling in particular alters low probabilities $\hat{p}(1\vert x)$ and results, if plotted against the non-resampling classifier $\hat{p}$, in a concave curve strictly above the line defined by the identity $\hat{p}^S(1\vert x)=\hat{p}(1\vert x)$.\\ 
    Let us assume $\hat{p}(1\vert x)$ to be consistent. In the special case of perfectly balanced data, i.e. $\hat{p}^S(y=1)=\tfrac{1}{2}$,  with an original base rate $\hat{p}(y=1)$ and the confidence predictions $\hat{p}^S(1\vert x)$, we can then explicitly reconstruct unbiased confidence predictions by 
            \begin{align} \label{remark_resampling_bias_correction}
                \hat{p}(1\vert x) &= \hat{p}(y=1)\left(\hat{p}(y=1)+\frac{(1-\hat{p}^S(1\vert x))(1-\hat{p}(y=1))}{\hat{p}^S(1\vert x)}\right)^{-1}.
            \end{align}
    
    We conclude this section with remarks on resampling.
    \begin{remark}
        \begin{itemize}
            \item[i)]  In a rare event setting the event probability $p(1\vert x)$ is generally low, leading to a small region in the feature space $X(\Omega)$ to be classified as the minority event. Lemma \ref{lemma:bias_resampling} and \ref{lemma:skew_resampling} illustrate that resampling increases this region and, thus, allows to shift the decision boundary of a classifier. For label predictions, this has empirically been shown to improve frequentist metrics, see \cite{Chawla.2002, Weiss.2001}.
            \item[ii)] If the given data represents the ground truth in the sense that data $\lbrace (x_i,y_i)\rbrace_{i=1}^{N}$ are i.i.d. realizations from $(X,Y)$ that allow for a sound estimation, in the sense of $\mathbb{E}_{x\sim X}\left[ p(1\vert x)-\hat{p}(1 \vert x)\right] = 0$, then resampling introduces a bias to our confidence predictions $\hat{p}^S(y\vert x)$. Any risk measure applied to the distribution induced by $\hat{p}^S(y\vert x)$ will reflect this bias.
            \item[iii)] It is obvious that random-undersampling disregards most of the information contained in data of the majority class. A natural approach to mitigate the loss of information is to use bootstrapping to obtain multiple balanced, random-undersampled training data and apply ensemble techniques, which have shown superior performance e.g. in \cite{Wallace.11.12.201114.12.2011}. 
        \end{itemize}
    \end{remark}
\section{Data} \label{section:data}
There exists a large number of empirical work that aims to describe the nature of the surrender probability $p(1\vert x)$, where $x$ includes macroeconomic, microeconomic and contractual factors, see \cite{Aleandri.2017, Burkhart.2018, Milhaud.2011, Yu.2019, Kiesenbauer.2012, Eling.2014}. As the nature of $p(1\vert x)$ seems to vary between countries and types of policy, we do not aim to construct an optimal estimator $\hat{p}(1\vert x)$ for all life insurance business. Our experiments focus on endowment contracts only. Further, we look at individual findings of the literature to construct four latent surrender models $p$ of varying nature and complexity, which we call \textit{surrender profiles}. Then, for each surrender profile we estimate $p$ with a particular focus on the consistency of estimator $\hat{p}$ and the effect of resampling. In the following, we summarize how data $\lbrace(x_i,y_i)\rbrace_{i=1}^N\sim(X,Y)$ is generated.


\paragraph{Simulation of contracts $x$.} Each endowment contract $x_i=(x_{i}^{(1)},\ldots,x_{i}^{(n)})$ is identified by the current age of the policyholder, the face amount, the duration of the contract, the elapsed duration, the remaining duration, the annual premium amount, the frequency of premium payments and the current calendar year. We start by generating a portfolio of $N_{\textcolor{\editcolor}{0}}=30'000$ endowment contracts at calendar year $0$, where we infer a realistic distributions of $X$ from \cite{Milhaud.2018}. In \cite{Milhaud.2018}, the authors provide detailed statistics on a portfolio of US whole life contracts. Given the similarity between endowment and whole life insurance\footnote{A whole life insurance is equivalent to an endowment insurance with infinite duration.}, the portfolio in \cite{Milhaud.2018} provides a realistic basis for our experiments. In Figure \ref{fig:portfolio_dist} we provide an illustration of the marginal distributions of the portfolio at calendar year $0$. Face amounts are calibrated such that annual premiums are consistent with \cite{Milhaud.2018}. To compute premiums we apply the equivalence-principle, see \cite{Fuhrer.2006, dickson_hardy_waters_2009}, under the following assumptions.
\begin{figure}[b!]
    \centering
    \includegraphics[width=\textwidth]{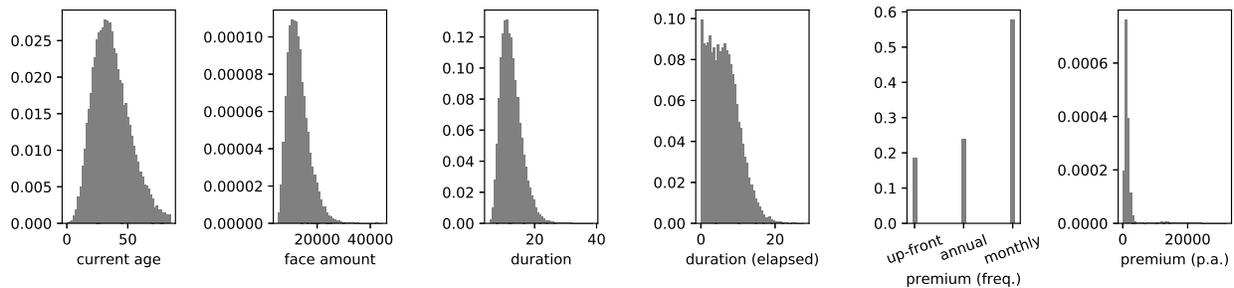}
    \caption{Marginal distribution of the portfolio at calendar year $0$.}
    \label{fig:portfolio_dist}
\end{figure}
\begin{itemize}
    \item The actuarial interest rate $i$ is constant with $i=0.02$.
    \item Expenses of the insurer for acquisition, administration and amortization of each contract are represented by $(\alpha, \beta, \alpha^{\gamma})=(0.025,0.03,0.001)$, see \cite{Burkhart.2018}.
    \item Premiums are paid up to the age of 67, the German age of retirement\footnote{To be precise, the age of retirement in Germany varies between 65 and 67 based on the date of birth. The basic retirement at 67 applies for individuals born on 1 January 1964 or later, see \cite{FED_Retirement_Germany}.}. Up-front premium payments are annualized linearly by the remaining time of premium payments.
    \item Mortality is described by a parametric survival model based on the Makeham Law, see \cite{dickson_hardy_waters_2009}. The $t$-year survival probability ${}_t p_a$ of individual aged $a$ is defined by 
    \begin{align} \label{eq_SUSM}
        {}_t p_a := \exp\left( -\text{A}t-\frac{\text{B}}{\log(\text{c})}
                        \text{c}^a(\text{c}^t-1) \right),
    \end{align}
    where we adopt the parameters of \cite{dickson_hardy_waters_2009} by setting the baseline hazard A$=0.00022$ and age related factors B$=2.7\cdot 10^{-7}$ and c$=1.124$. 
\end{itemize}

 More details on the simulation of data at calendar year $0$ can be found in the Appendix \ref{appendix:simulation_policies}. All \textcolor{\editcolor}{non-terminated} contracts are then iterated forward by increasing the age of the policyholder and the elapsed duration of the contract by the period length of one year. All other features are assumed to remain constant, i.e. changes to policies are not admissible. \textcolor{\editcolor}{The assumptions on when contracts terminate will be described in the subsequent paragraphs.} Additionally, at every iteration we simulate new business at the magnitude of $6\%$ of the existing business, which was empirically observed in the German market, see \cite{GDV.2019}. Contracts of the new business are simulated analogously to the portfolio at calendar year $0$, naturally with the restriction that the elapsed duration equals zero. \textcolor{\editcolor}{Over a time horizon of $15$ years we thereby generate about $N\in [260'000, 300'000]$ single year observations, depending on the specific assumptions, i.e. the surrender profile. The exact numbers and the imbalance is reported in the Appendix \ref{appendix:tables}, see Table \ref{table:data_review} \textcolor{\editcolor}{with $N=\vert\mathcal{D}_{\text{train}}\vert+\vert\mathcal{D}_{\text{test}}\vert$}. The next paragraphs describe the simulation of surrender events and introduce various surrender profiles.}

\paragraph{Meta model $p(1\vert \cdot)$.}

Given the input variable $x=(x^{(1)},\ldots,x^{(n)})$, we use a logistic regression model to obtain the true surrender probability $p(1\vert x)$ by
\begin{align}
    p(1\vert x) := \left(1+\exp(-\beta_0-\beta_x^{'}x)\right)^{-1},
\end{align}
where the regression coefficients $\beta_0\in\mathbb{R}$ and $\beta_x: \mathbb{R}^n\rightarrow\mathbb{R}^n$ are specified by the respective surrender profile. In contrast to its standard formulation, see \cite{Hastie.2017}, we denote coefficients $\beta_x$ as a function of $x$. We do so purely for notational convenience, as it allows us to set the effect $\beta_{x}^{(i)} x^{(i)}$ of the $i$th feature as piecewise constant without having to introduce hot-encoded auxiliary variables. 
This will be useful when describing our surrender profiles. Additionally, we note that the use of categorized risk drivers increases the complexity of the model $p(1\vert x)$, which is in line with the complexity of surrender activities.

\begin{figure}[bt] 
    \centering
    \begin{subfigure}{\textwidth}
        \centering
        \includegraphics[width=\textwidth]{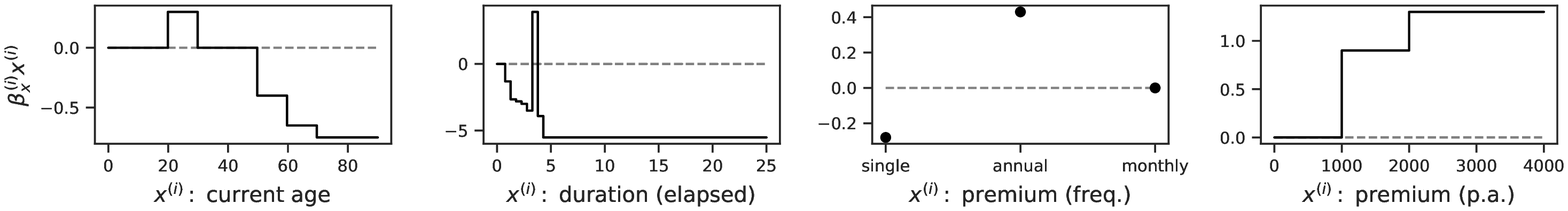}
        \subcaption{Surrender profile $1$}
        \label{fig:surrender_profile_1}
    \end{subfigure}
    \begin{subfigure}{\textwidth}
        \centering
        \includegraphics[width=\textwidth]{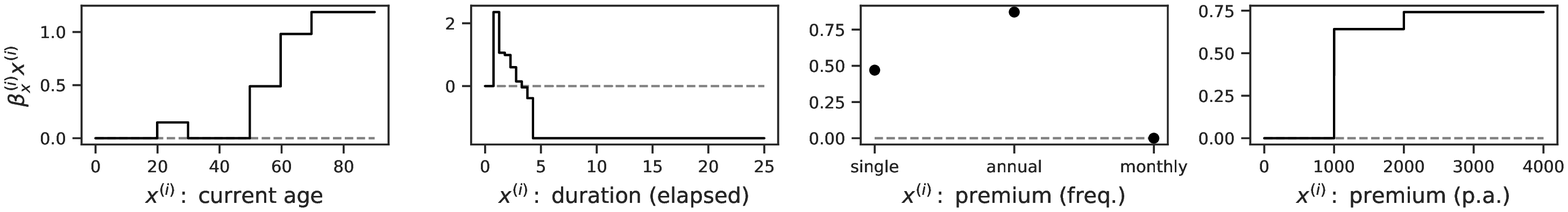}
        \subcaption{Surrender profile $2$}
        \label{fig:surrender_profile_2}
    \end{subfigure}
    \begin{subfigure}{\textwidth}
        \centering
        \includegraphics[width=\textwidth]{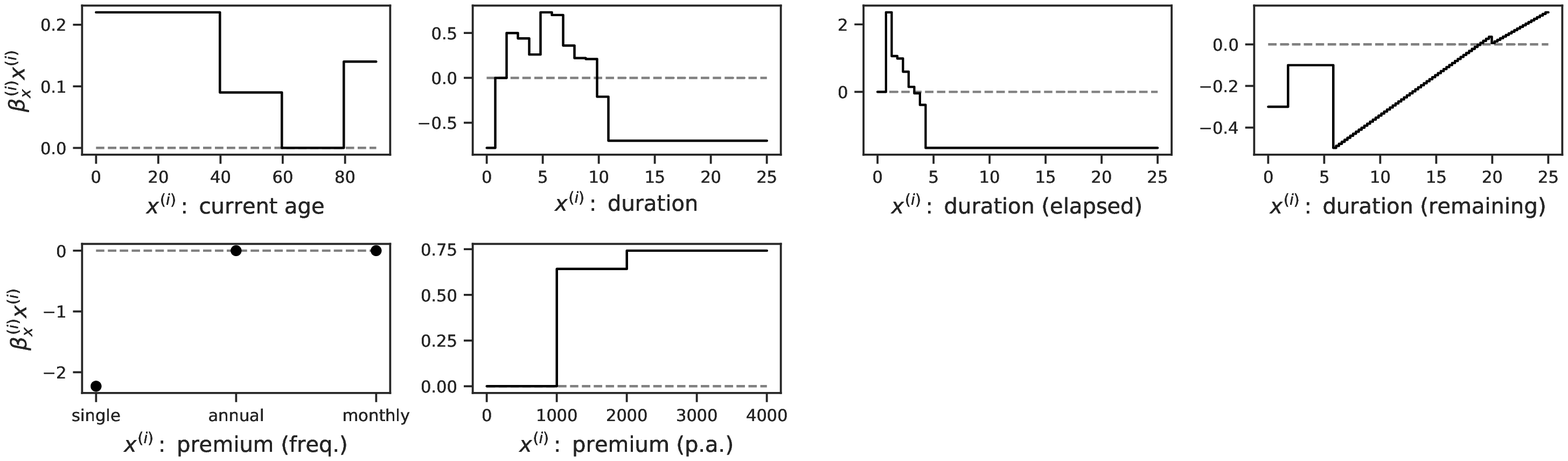}
        \subcaption{Surrender profile $3$}
        \label{fig:surrender_profile_3}
    \end{subfigure}
    \begin{subfigure}{\textwidth}
        \centering
        \includegraphics[width=\textwidth]{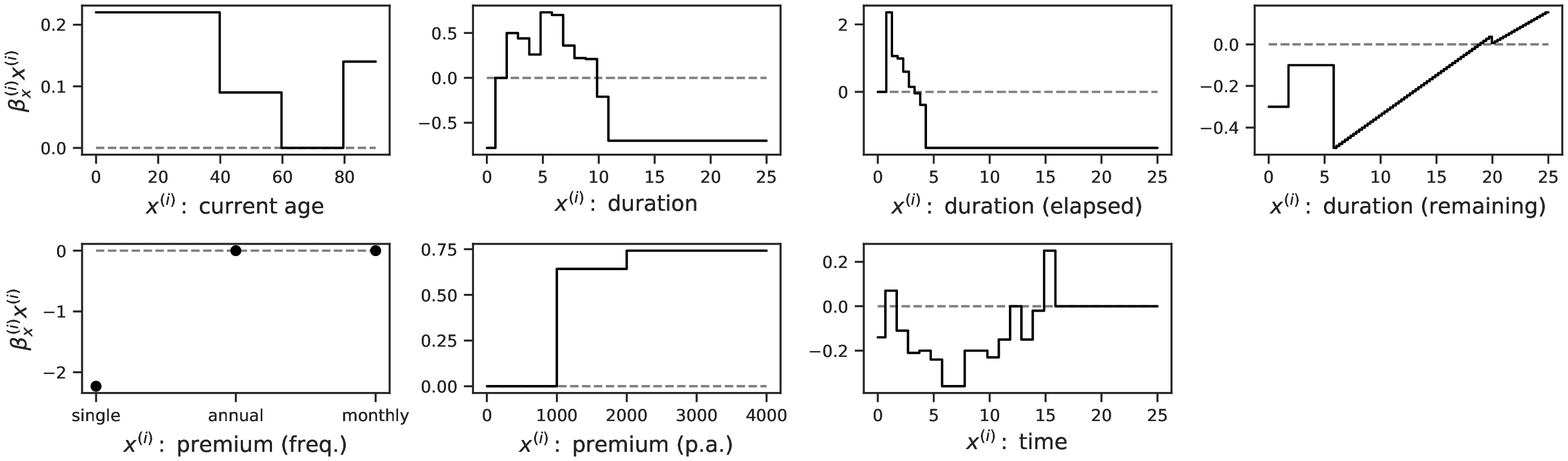}
        \subcaption{Surrender profile $4$}
        \label{fig:surrender_profile_4}
    \end{subfigure}
    \caption{Collection of distinct surrender profiles in our experiments.}
    \label{fig:surrender_profiles}
\end{figure}
In \cite{Milhaud.2011, Cerchiara.2008, Eling.2014}, the authors report specific values of $\beta_{x}^{(i)} x^{(i)}$ for a variety of categorized risk drivers $x^{(i)}$. We adopt and combine these values in four different surrender profiles, see Figure \ref{fig:surrender_profiles} for an illustration.  We consider the current age of a policyholder, the duration elapsed since underwriting the contract, the maximum duration of the contract, as well as the remaining duration until maturity, the frequency of premium payments and the annualized premium amount. \textcolor{\editcolor}{Most features have a locally constant effect on surrender, which translates to assuming that there exist of subgroups with homogeneous surrender behaviour. Also, from an approximation perspective locally constant functions are advantageous, as these simple, piece-wise constant functions can be viewed as a building block for estimating more intricate functions. } In surrender profile 4, the most complex setting, we additionally include a time factor represented by calendar years, which can be interpreted as a proxy of the economic environment at that time or alternatively as noise.  The range of profiles covers a variety of plausible dynamics, for example elevated surrender of young or higher ages, increased surrender at an early stage of contracts, for short-term contracts in general or higher surrender for higher amounts of annual premiums.  The intercept $\beta_0$ of each profile has been adjusted, such that observed surrender rates per calendar year fall into realistic ranges of 0.01 up to 0.05. In the Appendix \ref{appendix:surrender_profiles} we provide additional detail on the formulation of all surrender profiles.

\paragraph{Simulation of surrender activity $y$.} 
We observe a contract $x$ only up to its termination $T$, in our case either maturity, death or surrender. \textcolor{\editcolor}{We introduce the (potentially  censored) times of termination for all competing risks by $T^{(i)}:\Omega\rightarrow\mathbb{N}_{\geq 0}$, $i=1,2,3$. This is a slight misuse of notation. A more formal notation would include the earliest termination as a censoring on \textit{all} $T^{(i)}$. The reason for doing so is that we eventually simulate all $T^{(i)}$ separately and record the period and the type of earliest observation in $T=\min_{i} T^{(i)}$ and $J_T=\arg\!\min_i T^{(i)}$. } \\
Maturity events are easily identified by the duration \textcolor{\editcolor}{$d$} of a contract, \textcolor{\editcolor}{i.e. $T^{(3)}=d$}. To simulate the death of a policyholder with current age $a$ we numerically invert ${}_tp_a$ as a function of $t$ and obtain the time of \textcolor{\editcolor}{death} by $\left({}_{\textcolor{\editcolor}{u}}p_a\right)^{-1} \textcolor{\editcolor}{\in\mathbb{R}_{\geq 0}}$, $u\sim\mathcal{U}(0,1)$. \textcolor{\editcolor}{While $T^{(2)}$ only records the year of death, we also record the exact, non-integer time $({}_up_a)^{-1}$ to check which competing event is observed first.} Naturally, \textcolor{\editcolor}{the time of death is simulated} only once per policyholder. More detail on inverse transform sampling can be found in \cite{Glasserman.2010}. Surrender events are computed iteratively. \textcolor{\editcolor}{In simulation, the year of surrender is given by $T^{(1)}=\inf\lbrace t\in\mathbb{N}:~p(1\vert x_{t-1})\geq v \rbrace$ for $v\sim\mathcal{U}(0,1)$, where $x_t$ indicates the progression of a contract over time $t$}. 
If surrender and an alternative termination event are modeled to occur in the same year\textcolor{\editcolor}{, e.g. $T^{(1)}=5$ and $({}_up_a)^{-1}=4.2$}, we assume uniformity of the events, similar to fractional survival probabilities in \cite{Fuhrer.2006}. Given time $t\in(0,1)$ until the alternative termination event occurs\textcolor{\editcolor}{, e.g. $t=0.2$ as in the previous example}, we assume that the surrender event occurs with probability $t$ prior to the competing risk, i.e. maturity or death.

\paragraph{Data preparation.}
We one-hot encode the categorical feature 'frequency of premium payments' and apply min-max scaling to scale all contracts $x$ to the range of $[\textcolor{\editcolor}{-1},1]^n$. For the sake of readability, \textcolor{\editcolor}{we denote} the raw and the scaled contract both by $x$. Next, we split our data $\lbrace (x_i,y_i) \rbrace_{i=1}^{N}$ in time at the calendar year at which at least $70\%$ of data has been observed. Depending on the specific surrender profile, this split occurs at calendar years $t\in\lbrace 8,9 \rbrace$ and results in training data $\mathcal{D}_{\text{train}}$ and test data $\mathcal{D}_{\text{test}}$. Note that a split in time imitates a practical setting, where we test the model from a going concern perspective. A random train-test split would result in training on future data. Moreover, a split in time enables us to detect a potential bias of our model as the decomposition of the portfolio changes. Although we model new business to follow a stationary distribution, respectively all features but elapsed duration and the current calendar year, the inclusion of old business causes the composition of the portfolio to be non-stationary, e.g. with an increasing share of older policyholders or older contracts.\\ 
In the Appendix \ref{appendix:tables}, see Table \ref{table:data_review}, we additionally provide a summary of the number of one-year observations in training and test set, the imbalance of surrender and non-surrender events and the size of the training data set after random-undersampling or SMOTE-resampling to perfect balance. Without resampling, surrender events contribute a share of $0.01$ up to $0.045$ of all observations, which highlights the rare event character of our setting.

\section{\textcolor{\editcolor}{Modeling approaches}} \label{section:models} 
\textcolor{\editcolor}{For our analysis,} we \textcolor{\editcolor}{investigate} three commonly used model types, the logistic regression, tree based classifiers and neural networks, in their bagged and boosted form. \textcolor{\editcolor}{In theory, all three models have the capacity to fit the underlying, true surrender profiles arbitrarily close. While the logistic regression - given proper data preprocessing - coincides with the meta-model, standard tree-based classifiers by definition model locally constant probabilities. This implicit prior assumption matches our latent surrender profiles. Furthermore, neural networks serve as universal approximators, see \cite{Cybenko1989}, and can fit a large number of linear regions, see e.g. \cite{Montufar.2014}. Therefore, they are able to capture the dynamics of $p(1\vert x)$. At the same time, the optimization of a neural network is in general NP-hard, see e.g. \cite{Blum.1992}. So the question will be whether in our setting the proposed models can be optimized to identify the underlying policyholder behaviour.} As a substitute for a thorough exploratory data analysis, for each surrender profile we restrict the input $x$ to the estimator $\hat{p}(1\vert x)$ to the components $x_i$ that actually are part of the latent surrender model $p(1\vert x)$. \textcolor{\editcolor}{For the logistic regression we additionally provide explicit, polynomial feature engineering, to give the model more flexibility. Redundant features may be dropped due to an optional $L_1$-regularization. Hierarchical models such as neural networks or classification trees include implicit feature engineering. } Let us briefly comment on the specifics of the estimators\textcolor{\editcolor}{, including the specific parametrization found via extensive hyperparameter search. At this point, we do not apply resampling to the data set. }

\paragraph{\textcolor{\editcolor}{Approach of the hyperparameter search.}} \textcolor{\editcolor}{All following choices of hyperparameters are found by applying the package 'hyperopt', see \cite{bergstra2015hyperopt}, with its bayesian 'tpe.suggest' search algorithm and the cross-entropy as the loss function. For the logistic regression, random forest and XGBoost we take $128$ iterations and a $3$-fold cross-validation. Due to computational costs, for neural networks we do a simple split, where we set $30\%$ of the training data aside for validation, and decrease the number of iterations to $56$. This is justified as prior information of initial experiments, as e.g. practical depths and widths, allows us to reduce the search space. For all other models we primarily utilize the default space of hyperparameters suggested by \cite{bergstra2015hyperopt}. Specific changes include altering the maximum depth of a tree to a uniformly distributed integer value $d\in\lbrace 2,\ldots,20 \rbrace$ (random forest) and $d\in\lbrace 1,\ldots,11 \rbrace$ (XGBoost) or additional rates $r\in\lbrace 0.1, 0.2, \ldots, 0.7 \rbrace$ for the maximum number of features ('max$\_$features') considered during a split. More detail will be provided in the paragraph below on tree based methods.}

\paragraph{Baseline model.}
To set a baseline \textcolor{\editcolor}{performance}, we introduce a classifier with a constant surrender rate. This provides practical supervision whether subsequent classifiers are plausible and whether their increasing complexity is justified. Given a set of training data $\mathcal{D}_{\text{train}}$, which differs for each surrender profile, we define the confidence prediction $\hat{p}$ of the baseline model by
    \begin{align*}
        \hat{p}(1\vert x^{\prime}) &\equiv \frac{1}{\vert\mathcal{D}_{\text{train}} \vert}\sum_{(x,y)\in\mathcal{D}_{\text{train}}}y~~,
    \end{align*}
where $x^{\prime}$ presents an arbitrary contract.

\paragraph{Logistic regression.}
The structure of the logistic regression is given by $p(1\vert x) := \left(1+\exp(-\beta_0-\beta^{'}x)\right)^{-1}$, with intercept $\beta_0\in\mathbb{R}$ and coefficients $\beta\in\mathbb{R}^p$. In contrast to the meta model \textcolor{\editcolor}{in Section \ref{section:data}}, coefficients $\beta$ are constants. Providing the estimator with the true categories of all components of $x$, e.g. indicating ages 20-30 to be highly susceptible to surrender, seems unrealistic in a practical rare event setting. Instead, we feature engineer higher degree inputs $(x_{i}^{(j)})^k$, $k=1,\ldots,\textcolor{\editcolor}{4}$, for each feature $x_{i}^{(j)}$. \textcolor{\editcolor}{Note that for all surrender profiles the meta model for surrender does not consider interactions between features, which is why we do not engineer inputs of the form $(x_{i}^{(l)})^k(x_{i}^{(n)})^m$.} \textcolor{\editcolor}{Degrees of orders $k>4$ are not considered, as work in e.g. \cite{Weiss.2019, Krah.2018} suggests these to be linked with numeric instability and poor explanatory power. Also, additional features $(x_{i}^{(j)})^k$ with degrees $k>4$ would amplify multicollinearity in the data. } \textcolor{\editcolor}{Next, we} apply $L_{\textcolor{\editcolor}{p}}$-regularization $\sum_{i>0}\textcolor{\editcolor}{\vert}(\beta^{(i)}\textcolor{\editcolor}{\vert}^{\textcolor{\editcolor}{p}}$ \textcolor{\editcolor}{with $p\in\lbrace 1,2\rbrace $, where $L_1$-regularization provides implicit feature-selection.} \textcolor{\editcolor}{The inverse of the regularization strength ('$C$') and type of this regularization ('penalty') are then calibrated during a hyperparameter search and reported in Table \ref{table:hps:logit}. We observe little regularization for profiles 1 and 3, medium regularization for profile 2 and comparably high regularization for profile 4. Further, with the exception of profile 2, all profiles apply a $L_2$-regularization.} At last, we combine $N_{\text{bag}}=10$ estimators to an ensemble model\footnote{For completeness, we also tested adaptive boosting as in \cite{RobertE.Schapire.1999} on polynomial features, but dropped it due to poor performance.}, since a single logistic estimator is highly volatile in a rare event setting, see again \cite{Hastie.2017}. \textcolor{\editcolor}{Larger values of $N_{\text{bag}}$ showed no further improvements in our experiments.}  

\begin{table}[hptb]
    \centering
    \begin{tabular}{l| cc}
profile &            C &  penalty  \\
\midrule
1 &  1.26e+08 &         l2  \\
2 &      4.05 &          l1  \\
3 &  8.22e+06 &          l2  \\
4 &     0.22 &         l2  \\
\bottomrule
\end{tabular}

    \caption{\textcolor{\editcolor}{Hyperparameters for logistic regression, in the notation of 'sklearn' \cite{scikit-learn}.}}
    \label{table:hps:logit}
\end{table}

\paragraph{Tree based classifiers.}
A tree based classifier provides a split of the feature space $X(\Omega)$ into disjoint regions $R_1,\ldots,R_m$. The surrender probability $\hat{p}(1\vert x)$ is then formed by relative frequencies per regions $R_k$. For an arbitrary data point $x^{\prime}\in R_k$ we look at the subset $\mathcal{D}_{R_k}=\lbrace (x,y)\in\mathcal{D}_{\text{train}} : x\in R_k \rbrace$ and assign
        \begin{align*}
            \hat{p}(1\vert x^{\prime}) =  \frac{1}{\vert \mathcal{D}_{R_k} \vert}\sum_{(x,y)\in \mathcal{D}_{R_k}} y~~.
        \end{align*}
For an individual classification and regression tree (CART), see \cite{Hastie.2017}, regions $R_k$ are the result of recursive, binary splits of the region $X(\Omega)$. At each recursion, a CART searches for the feature $x^{(i)}$ and the hyperplane $H:=\lbrace x:~x^{(i)}=c \rbrace$, $c\in\mathbb{R}$, such that the respective split by $H$ leads to a maximum reduction of the impurity. We measure the impurity of a region $R$ by the binary cross-entropy\footnote{We note that the alternative \textit{Gini index}, see \cite{Hastie.2017}, yielded comparable results in our experiments.}. \textcolor{\editcolor}{Next, we use ensemble techniques and form a random forest classifier as e.g. in \cite{Hastie.2017} and the XGBoost classifier, a gradient boosting decision tree, see \cite{Chen.2016}. Both models are subject to a hyperparameter search for all surrender profiles.} \\
\textcolor{\editcolor}{For the random forest we tune the number of trees ('n$\_$estimators'), the maximum depth of each tree ('max$\_$depth'), how many samples are required for an iteration of splitting ('min$\_$samples$\_$split') and how many samples at least have to lie in the final layer of each tree ('min$\_$samples$\_$leaf'). Further, we optimize which fraction ('max$\_$features') of the features is considered during the splitting and whether to train individual trees on bootstrapped training data ('bootstrap'). Results in Table \ref{table:hps:rf} show medium-sized individual trees, which consider roughly $50\%$-$70\%$ of the features, when looking for an optimal split. Notable exceptions are profile 2, where all features are considered, when looking for a split, and profile 1, where no bootstrapping is applied. As the underlying surrender profiles become more complex, we see the number of individual trees increase. For the hyperparameters 'min$\_$samples$\_$leaf' and 'min$\_$samples$\_$split' we found the default values to work best, which yield little regularization.} 

\begin{table}[hptb]
    \centering
    \begin{tabular}{l|cccccc}
profile & bootstrap & max\_depth & max\_features & min\_samples\_leaf & min\_samples\_split & n\_estimators \\
\midrule
1 &     False &         7 &          0.6 &                1 &                 2 &          336 \\
2 &      True &         5 &         1 &                1 &                 2 &          403 \\
3 &      True &         8 &          0.7 &                1 &                 2 &         1303 \\
4 &      True &         8 &          0.5 &                1 &                 2 &         1202 \\
\bottomrule
\end{tabular}

    \caption{\textcolor{\editcolor}{Hyperparameters for random forest, in the notation of 'sklearn' \cite{scikit-learn}.}}
    \label{table:hps:rf}
\end{table}

\textcolor{\editcolor}{
In a similar fashion, for XGBoost  we look at the number of trees ('n$\_$estimators'), the maximum depth of individual trees ('max$\_$depth'), the strength of $L_2$- and $L_1$-regularization ('reg$\_$lambda' and 'reg$\_$alpha') and which portion of the data individual weak learners should be formed on ('subsample'). Further, we look at the fraction of features considered during splitting trees ('colsample$\_$bylevel'). The parameters 'gamma', 'learning$\_$rate' and 'min$\_$child$\_$weight' provide additional, regularizing parameters relating to the minimum loss reduction, a shrinkage factor per step and a minimum measure of importance to justify a split. For more detail we refer the reader to the documentation of \cite{Chen.2016} in the respective python package. \\
The results are presented in Table \ref{table:hps:xgb}. Most notably, we see very shallow tree structures of depths 1 up to 3 with relatively many individual trees. The exception of profile 3 with only 320 individual trees shows a comparably strong $L_1$-regularization. The shallow structures seem to be induced by the minimum weights required during splitting, i.e. 'min$\_$child$\_$weight', which are much higher than the default value of $1$. Lastly, profiles 1 and 4 show very conservative and low learning rates, in comparison to the default value of $0.3$.  }

\begin{table}[hptb]
    \centering
    \begin{minipage}{0.75\textwidth}
    \begin{flushleft}
        \begin{tabular}{l|cccccccccc}
        profile&  colsample\_bylevel &       gamma &  learning\_rate &  max\_depth &  min\_child\_weight  \\
        \midrule
        1 &           0.80 &            4.58 &       0.01 &        3 &              10 &          \\
        2 &           0.52 &            1.45 &       0.30 &        1 &              70 &          \\
        3 &           0.62 &            1.72 &       0.33 &        1 &              51 &          \\
        4 &           0.99 &            2.57 &       0.03 &        2 &              10 &          \\
        \bottomrule
             \textcolor{white}{profile} &    n\_estimators &  reg\_alpha &  reg\_lambda &  subsample & &\\
            \midrule
        1     &                   920 &   0.00 &    1.52 &   0.71 & & \\
        2    &                    980 &   0.00 &    1.07 &   0.99 & &\\
        3   &                    320 &   0.74 &    1.37 &   0.96 & &\\
        4  &                    780 &   0.023 &    1.00 &   0.64 & &\\
        \bottomrule
        \end{tabular}
    \end{flushleft}    
\end{minipage}

    \caption{\textcolor{\editcolor}{Hyperparameters for XGBoost.}}
    \label{table:hps:xgb}
\end{table}

\paragraph{Neural network classifiers.}
\textcolor{\editcolor}{We again consider a bagged and a boosted architecture. All models will be trained with stochastic gradient descent on the cross-entropy loss, implemented by the 'adam' algorithm, see \cite{adam}. To mitigate overfitting and reduce runtime, we apply early stopping\textcolor{\editcolor}{, which aborts training} if the validation loss does not decrease for 25 epochs. We explore batch sizes per GPU\footnote{The training of neural models is distributed across eight Quadro RTX 8000.} of $64, 128$ and $256$  and learning rates of $10^{-i}$ with $i=2, 2.5, \ldots, 4$ for bagging and $i=0.5, 1, 1.5, 2$ for boosting, respectively. The range of learning rates is motivated by preliminary experiments. }\\ 
The bagged model consists of neural networks $\hat{p}^{(k)}(1\vert \cdot)$, with $k=1,\ldots,N_{\text{bag}}$, \textcolor{\editcolor}{and $N_{\text{bag}}=5$.} All models $\hat{p}^{(k)}$ have the same architecture. We allow feed-forward network architectures \textcolor{\editcolor}{with up to 3 hidden layers, each with width$_i\in \lbrace 10, 15, \ldots, 50\rbrace$ units, sigmoid or ReLU activation ('actv') and dropout between hidden layers, which is uniformly, randomly sampled with a rate of at most $0.5$.}. The output layer has a single unit and a sigmoid activation $\sigma$. The bagged estimator is then formed by $\hat{p}(1\vert x):= \tfrac{1}{N_{\text{bag}}}\sum_k \hat{p}^{(k)}(1\vert x)$. \textcolor{\editcolor}{The search for optimal hyperparameters is performed for a single model only, since the models $\hat{p}^{(k)}(1\vert \cdot)$ are trained independently. However, to improve model performance, we perform a final fine tuning, where we collective train $\hat{p}(1\vert x)$, as done e.g. in \cite{Kiermayer2020}. For this step we reduce the learning rate by half. \\
In Table \ref{table:hps:ann:bag} we report the choice of hyperparameters for the architecture and common quantities for the stochastic gradient descent training performed by the 'adam' algorithm, see \cite{adam}. We observe the models for all surrender profiles to prefer the ReLU activation and, with the exception of profile 2, a depth of 3 hidden layers is superior to shallower structures. Note that the number of linear regions in $\hat{p}^{(k)}$ grows exponentially in the depths and only polynomially in the width, see e.g. \cite{Montufar.2014}. Further, it is notable that for profile 2 and 4 significantly higher dropout rates are preferred, which generally can be thought of as stronger regularization in favour of improved generalization, see e.g. \cite{sriv.2014}. For profile 4 we also find a comparably lower batch size for training to be beneficial. The individual widths of layers, as well training parameters are rather homogeneous and within the selected range.} \\

\begin{table}[htb]
    \centering
    \begin{tabular}{l|cccccccc}
profile &  actv & batch\_size & depth &    dropout &       lrate  & width\_1 & width\_2 & width\_3 \\
\midrule
1 &  ReLU &        256 &     3 &  0.007 &  0.003 &      35 &      25 &      20 \\
2 &  ReLU &        128 &     2 &   0.208 &       0.001 &      50 &      20 &      n/a \\
3 &  ReLU &        128 &     3 &   0.034 &  0.003 &      20 &      25 &      50 \\
4 &  ReLU &         64 &     3 &   0.123 &       0.001 &      20 &      20 &      20 \\
\bottomrule
\end{tabular}

    \caption{\textcolor{\editcolor}{Hyperparameters for a bagged ensemble of neural networks. Abbrev.: activation function: 'actv', learning rate: 'lrate', number of hidden layers: 'depth', width of hidden layer i: 'width$\_$i'.}}
    \label{table:hps:ann:bag}
\end{table}

Additionally, we construct a boosted neural network. Motivated by \textcolor{\editcolor}{preliminary} experiments and \cite{Badirli.2020}, we choose shallow neural networks $p^{(k)}$, $k=1,\ldots,N_{\text{boost}}$ as weak learners with $N_{\text{boost}}=\textcolor{\editcolor}{5}$. \textcolor{\editcolor}{The choice of $N_{\text{boost}}$ is empirically motivated, as we observed no benefit of larger values.} For $k>1$, each network $p^{(k)}$ has one hidden layer with \textcolor{\editcolor}{a width of up to $200$ units\footnote{For computational reasons, we use only 28 trials for one round of hyperparameter tuning of the boosted network and choose a stratified approach. We explore six different widths in steps of $10$ units at a time. Greater widths are then explored iteratively in follow-up searches of again 28 trials if the previous optimal parameter is found at the upper border of the search space.}}, a ReLU activation in the hidden layer and a single output unit with a linear activation. We start the boosting procedure with the baseline rate of $\hat{p}^{(1)}(1\vert x):=\sigma^{-1}\left(\tfrac{1}{\vert \mathcal{D}_{\text{train}} \vert}\sum_i y_i\right)$ for training data $(x_i,y_i)\in\mathcal{D}_{\text{train}}$. We then iteratively add and interdependently train weak learners $p^{(k)}$ until we obtain the final boosted model $\hat{p}:\mathbb{R}^p\rightarrow (0,1)$ by
\begin{align}\label{eq:NN_boost_prediction}
    \hat{p}(1\vert x) := \sigma\left(\sum_{k=1}^{N_{\text{boost}}} \hat{p}^{(k)}(1\vert x) \right).
\end{align}
\textcolor{\editcolor}{Similarly to the bagged setting and as in \cite{Badirli.2020}, we apply a final fine tuning for the collective estimator $\hat{p}(1\vert x)$, with a learning rate cut by half. Note that in contrast to \cite{Badirli.2020}, which utilizes gradient boosting and hence a Taylor approximation, in \eqref{eq:NN_boost_prediction} we work with the exact cross-entropy loss. While this is computationally more expensive as we cannot train on residuals, it allows us to use the sigmoid function $\sigma$ in \eqref{eq:NN_boost_prediction} to ensure sound confidence predictions $\hat{p}(1\vert x)\in(0,1)$. Lastly, the choice of hyperparameters is reported in Table \ref{table:hps:ann:boost}. The most notable differences are the much higher widths for profiles 1 and 4. For our purpose, single-layer, weak learners seem sufficient. For more elaborate surrender behaviour, e.g. with a higher dimensional feature space, multiple-layer learners might be beneficial, as suggested by \cite{Badirli.2020, Montufar.2014}. }
\begin{table}[hpbt]
    \centering
    \begin{tabular}{l|ccccc}
profile &  actv & batch\_size &      lrate &   width \\
\midrule
1 &  ReLU &         64 &  0.003 &             80 \\
2 &  ReLU &        256 &   0.031 &             40 \\
3 &  ReLU &        128 &        0.01 &             20 \\
4 &  ReLU &        128 &  0.003 &            120 \\
\bottomrule
\end{tabular}

    \caption{\textcolor{\editcolor}{Hyperparameters for a boosted ensemble of neural networks. Abbrev.: activation function: 'actv', learning rate: 'lrate'.}}
    \label{table:hps:ann:boost}
\end{table}

\section{\textcolor{\editcolor}{Numerical Results}} \label{section:num_experiments}
Let us now evaluate the classifiers mentioned above for all four surrender profiles. We start with a hypothetical analysis, where we assume that we can access the true probability of surrender within the next year of an arbitrary contract $x$. While this quantity $p(1\vert x)$ is not available in practice it yet provides instructive insight into the classifier's performance \textcolor{\editcolor}{and the information we would like to retrieve from observable quantities}. Also, it allows us to compare the bias of resampling to Lemma \ref{lemma:bias_resampling} and \ref{lemma:skew_resampling} numerically. We begin with a qualitative analysis of bias and variance of the models before we provide statistics to quantify our observations and the effect of resampling. In the end, we look at the modelled mean surrender rate over calendar-time, including confidence bands, for all classifiers and all profiles. In summary, we will see XGBoost \textcolor{\editcolor}{and random forest} to be the superior model \textcolor{\editcolor}{for the given data}, closely followed by neural classifiers. \textcolor{\editcolor}{We are also able to} confirm the bias of resampling\textcolor{\editcolor}{, which was derived in Section \ref{section:rare_events},} numerically. \textcolor{\editcolor}{Further, as we add non-stationary noise to the surrender behaviour, we will observe tree based methods to be more stable than our alternatives.} \textcolor{\editcolor}{Lastly}, we will notice that model evaluation by a single value, e.g. cross-entropy, is insufficient as it ignores the non-stationarity of our data. \textcolor{\editcolor}{Therefore}, mean surrender rates with confidence bands as in Corollary \ref{cor:confidence_intervals} will provide a more comprehensive perspective.

\paragraph{\textcolor{\editcolor}{Evaluation with complete information.}} In Figure \ref{fig:pq_plots} we consider test data $(x,y)\in\mathcal{D}_{\text{test}}$ and plot the label predictions $\hat{p}(1\vert x)$ against the true, latent surrender probability $p(1\vert x)$ for all discussed models and all four surrender profiles.

\begin{figure}[htbp] 
        \centering
        \begin{subfigure}{.9\textwidth} 
            \centering
            \includegraphics[width=.95\textwidth]{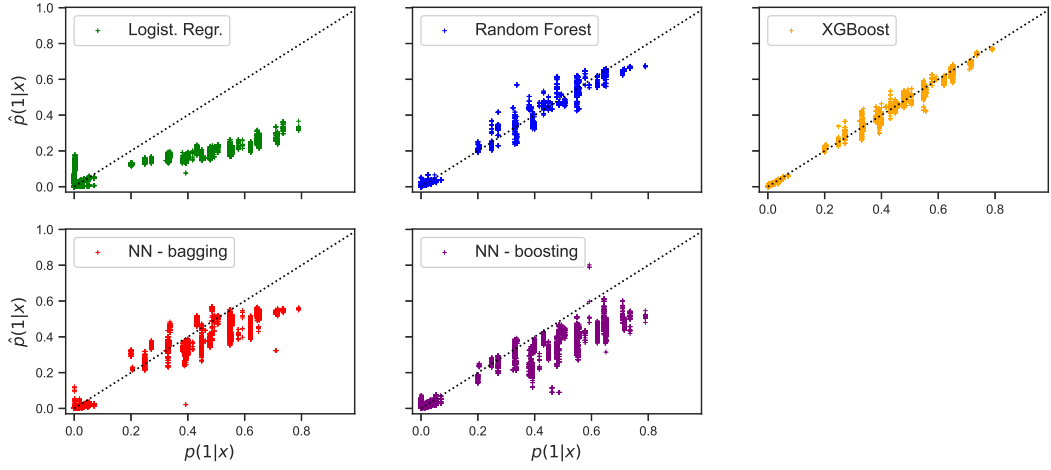}
            \subcaption{Profile 1}
            \label{fig:pq_1}
        \end{subfigure}
        \centering
        \begin{subfigure}{.9\textwidth} 
            \centering
            \includegraphics[width=.95\textwidth]{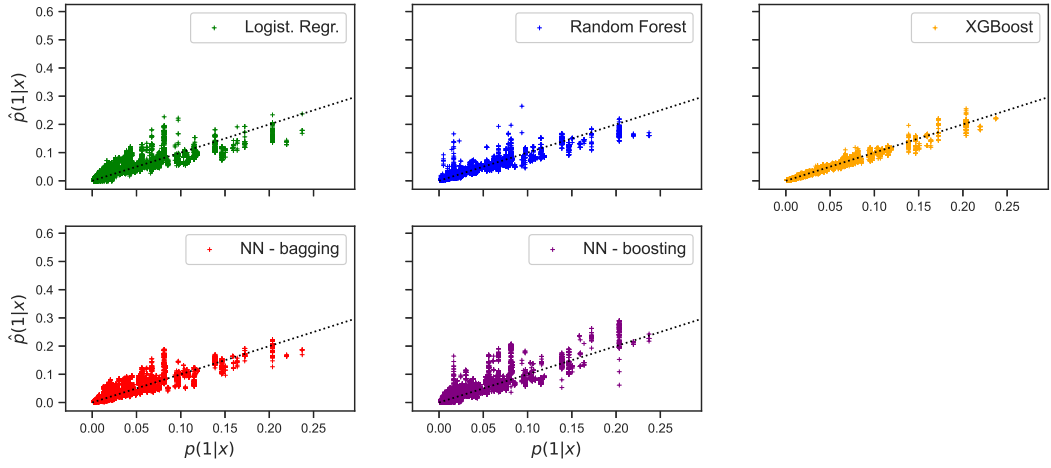}
            \subcaption{Profile 2}
            \label{fig:pq_2}
        \end{subfigure}
        \centering
        \begin{subfigure}{.9\textwidth} 
            \centering
            \includegraphics[width=.95\textwidth]{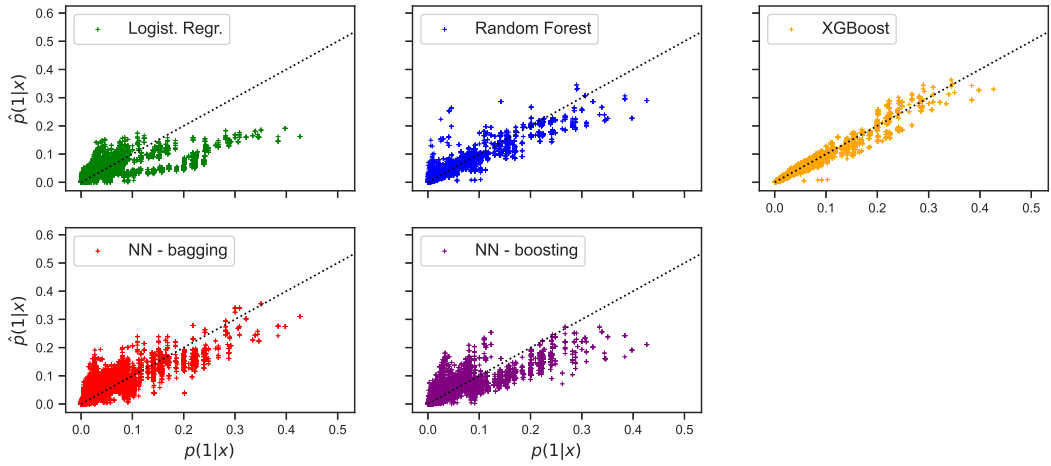}
            \subcaption{Profile 3}
            \label{fig:pg_3}
        \end{subfigure}
        \caption{$p$-$\hat{p}$ plots of the discussed classifiers for surrender profiles 1-4  \textcolor{\editcolor}{for $(x,\cdot)\in\mathcal{D}_{\text{test}}$}. }
\end{figure}
\begin{figure}[htbp]\ContinuedFloat
        \centering
        \begin{subfigure}{.9\textwidth} 
            \centering
            \includegraphics[width=.95\textwidth]{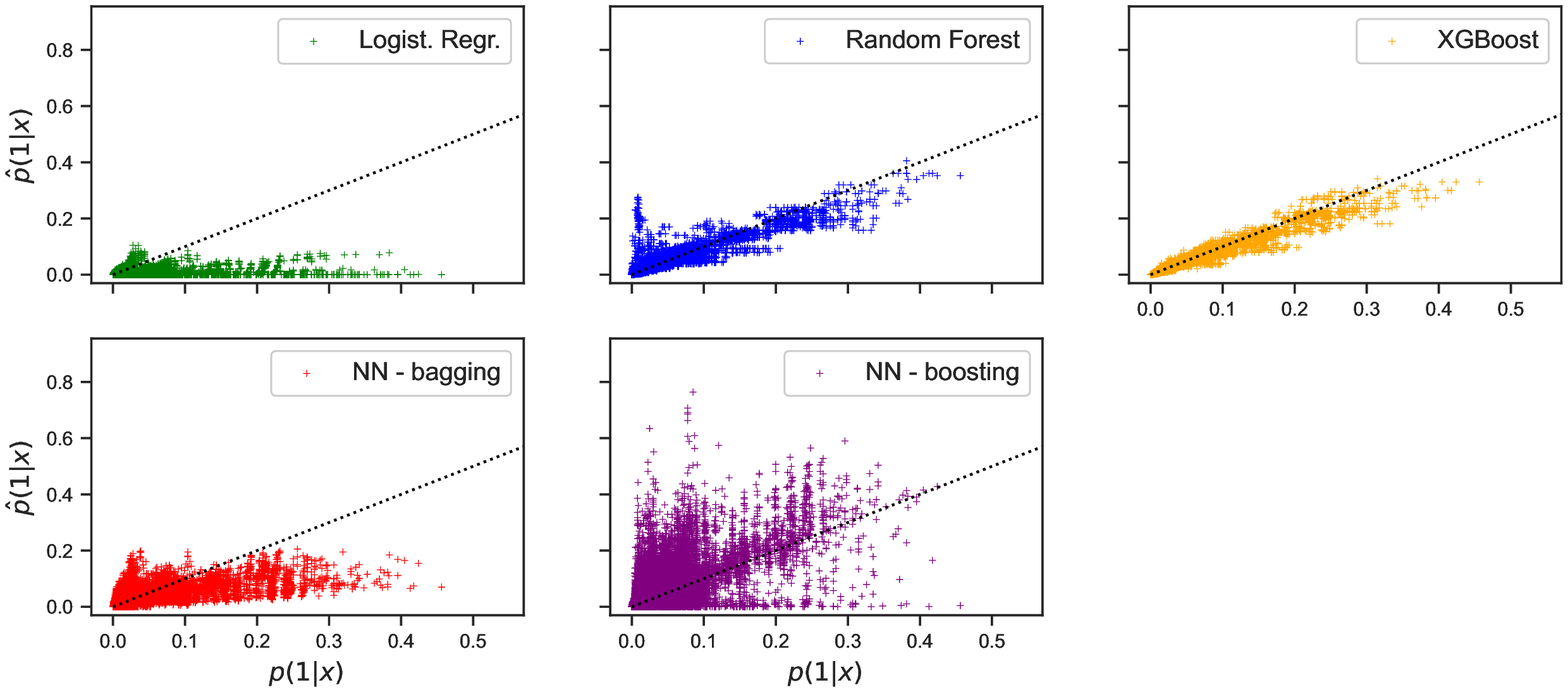}
            \subcaption{Profile 4}
            \label{fig:pq_4}
        \end{subfigure}
    \captionsetup{list=off,format=cont}
    \caption{$p$-$\hat{p}$ plots of the discussed classifiers for surrender profiles 1-4  \textcolor{\editcolor}{for $(x,\cdot)\in\mathcal{D}_{\text{test}}$}. }
    \label{fig:pq_plots}
\end{figure}

We observe that XGBoost \textcolor{\editcolor}{and random forest generally perform best and are approximately unbiased for all surrender profiles, i.e. their deviations from the 45°-line in Figure \ref{fig:pq_plots} are roughly symmetric.  Also, in comparison to other models, the errors made by XGBoost seem to exhibit a lower variance. This indicates that the respective bins of locally constant policyholder behaviour have successfully been identified. For profiles 1 and 3 non-tree-based methods show a slight bias and underestimate the surrender probability $p(1\vert x)$ for high risk contracts $x$. On the contrary, low risk contracts $x$ in profile 3 seem on average to be overestimated by non-tree-based methods. Recall that two lapse events $y_1,~y_2$ are realizations of different random variables $(Y\vert X=x_i)$ \textcolor{\editcolor}{with $i=1,2$}, given that the underlying contracts $x_1,x_2$ do not coincide. Hence, large or small values of $p(1\vert x)$ indicate the risk level of a specific contracts and not the tail of a single distribution. Further, neural networks seem to outperform the logistic regression, arguably since its implicit, hierarchical feature engineering is more adaptive than selective, polynomial feature engineering provided to the logistic regression\footnote{\textcolor{\editcolor}{Increasing the maximum degree of polynomially engineering features from 4 to 10, at the risk of numerical instability and multicollinearity, did generally not show any significant improvement for the logistic regression.}}. In particular in rare event settings, where little information on dominating features is accessible, this seems to be a desirable property. In profile 2, all models appear to be able to identify the dynamics of surrender approximately equally, with only slight differences in terms of the variance. The most distinct differences are found in profile 4, which has the same characteristics as profile 3 with additional, non-stationary noise that reverses its trend as we move from training to test data, see Figure \ref{fig:surrender_profile_4}. This provides a particularly challenging setting, where classifiers are forced to extrapolate, since the distribution of training and test data increasingly differs over time. In this setting, we clearly observe the tree based methods to provide more stable estimates, whereas other methods show poor results, see Figure \ref{fig:pq_4}. The logistic regression fails to generalize well to the test data, arguably since the underlying dynamics of surrender exceed the capacity of the model. The calibrated neural networks either underestimate the true probability $p(1\vert x)$, as for the bagged ensemble, or provide a high variance estimator. Instead of trying to mitigate this poor generalization, e.g. by increasing the size of the validation data or the strength of regularization, we take profile 4 as an use case and aim to identify non-stationary behaviour of policyholders in our practical model evaluation.}\\
 Recall that the illustrated $p$-$\hat{p}$ plot cannot be used for tuning hyperparameters, as the true probabilities $p(1\vert x)$ are unknown. \textcolor{\editcolor}{Therefore, we will now try to retrieve this latent information in a practical workflow. For completeness, we also quantify our observations from Figure \ref{fig:pq_plots} on bias on variance of all models.}\\


\begin{table}[bp] 
        \centering
        \begin{subtable}{\textwidth} 
            \centering
            \begin{tabular}{l|ccc:cc|ccc:cc}
Data & \multicolumn{5}{l}{Train.} & \multicolumn{5}{l}{Test} \\ 
Evaluation &    acc. &     $F_1$ & entropy & mae & Var &    acc. &    $F_1$ & entropy & mae & Var \\
\midrule
Baseline      &  0.9710 &  0.0000 &       0.1311 &           0.0496 &       1.08e-02 &  0.9546 &  0.0000 &       0.1890 &           0.0642 &       1.69e-02 \\
Logist. Regr. &  0.9710 &  0.0000 &       0.0812 &           0.0340 &       6.39e-03 &  0.9546 &  0.0000 &       0.1005 &           0.0359 &       7.75e-03 \\
Random Forest &  0.9774 &  \textcolor{gray}{0.6234} &       \textcolor{gray}{0.0461} &           0.0025 &       1.13e-04 &  0.9633 &  \textcolor{gray}{0.6121} &       0.0708 &           0.0035 &       1.54e-04 \\
XGBoost       &  \textcolor{gray}{0.9777} &  0.6213 &       0.0481 &           \textcolor{gray}{0.0019} &       \textcolor{gray}{3.56e-05} &  \textcolor{gray}{0.9636} &  0.6112 &       \textcolor{gray}{0.0706} &           \textcolor{gray}{0.0025} &       \textcolor{gray}{5.27e-05} \\
NN - bagging  &  0.9743 &  0.5112 &       0.0560 &           0.0096 &       1.37e-03 &  0.9607 &  0.4115 &       0.0740 &           0.0102 &       9.85e-04 \\
NN - boosting &  0.9733 &  0.3092 &       0.0603 &           0.0150 &       2.36e-03 &  0.9563 &  0.1120 &       0.0766 &           0.0143 &       1.97e-03 \\
\bottomrule
\end{tabular}

            \subcaption{Profile 1}
            \label{table:stats1}
        \end{subtable}
        \centering
        \begin{subtable}{\textwidth}
            \centering
            \begin{tabular}{l|ccc:cc|ccc:cc}
Data & \multicolumn{5}{l}{Train.} & \multicolumn{5}{l}{Test} \\ 
Evaluation &    acc. &     $F_1$ & entropy & mae & Var &    acc. &    $F_1$ & entropy & mae & Var \\
\midrule
Baseline      &  0.9835 &  0.0 &       0.0842 &           0.0169 &       7.84e-04 &  0.9869 &  0.0 &       0.0703 &           0.0147 &       2.70e-04 \\
Logist. Regr. &  0.9835 &  0.0 &       0.0724 &           0.0072 &       3.79e-04 &  0.9869 &  0.0 &       0.0616 &           0.0054 &       7.59e-05 \\
Random Forest &  0.9835 &  0.0 &       \textcolor{gray}{0.0665} &           0.0037 &       1.21e-04 &  0.9869 &  0.0 &       0.0603 &           0.0027 &       2.70e-05 \\
XGBoost       &  0.9835 &  0.0 &       0.0691 &           \textcolor{gray}{0.0020} &       \textcolor{gray}{3.65e-05} &  0.9869 &  0.0 &       \textcolor{gray}{0.0599} &           \textcolor{gray}{0.0013} &       \textcolor{gray}{6.15e-06} \\
NN - bagging  &  0.9835 &  0.0 &       0.0713 &           0.0058 &       3.66e-04 &  0.9869 &  0.0 &       0.0607 &           0.0040 &       5.81e-05 \\
NN - boosting &  0.9835 &  0.0 &       0.0718 &           0.0065 &       3.88e-04 &  0.9869 &  0.0 &       0.0612 &           0.0047 &       7.83e-05 \\
\bottomrule
\end{tabular}

            \subcaption{Profile 2}
            \label{table:stats2}
        \end{subtable}
        \centering
        \begin{subtable}{\textwidth}
            \centering
            \begin{tabular}{l|ccc:cc|ccc:cc}
Data & \multicolumn{5}{l}{Train.} & \multicolumn{5}{l}{Test} \\ 
Evaluation &    acc. &     $F_1$ & entropy & mae & Var &    acc. &    $F_1$ & entropy & mae & Var \\
\midrule
Baseline      &  0.9799 &  0.0000 &       0.0984 &           0.0206 &       8.84e-04 &  0.9838 &  0.0 &       0.0831 &           0.0185 &       5.48e-04 \\
Logist. Regr. &  0.9799 &  0.0000 &       0.0868 &           0.0131 &       5.74e-04 &  0.9838 &  0.0 &       0.0734 &           0.0106 &       3.50e-04 \\
Random Forest &  \textcolor{gray}{0.9800} &  \textcolor{gray}{0.0059} &       \textcolor{gray}{0.0738} &           0.0045 &       1.09e-04 &  0.9838 &  0.0 &       0.0684 &           0.0036 &       5.92e-05 \\
XGBoost       &  0.9799 &  0.0014 &       0.0790 &           \textcolor{gray}{0.0024} &       \textcolor{gray}{2.86e-05} &  0.9838 &  0.0 &       \textcolor{gray}{0.0678} &           \textcolor{gray}{0.0020} &       \textcolor{gray}{1.92e-05} \\
NN - bagging  &  0.9799 &  0.0000 &       0.0816 &           0.0080 &       2.45e-04 &  0.9838 &  0.0 &       0.0696 &           0.0061 &       1.47e-04 \\
NN - boosting &  0.9799 &  0.0000 &       0.0824 &           0.0087 &       3.19e-04 &  0.9838 &  0.0 &       0.0698 &           0.0069 &       1.75e-04 \\
\bottomrule
\end{tabular}

            \subcaption{Profile 3}
            \label{table:stats3}
        \end{subtable}
        \begin{subtable}{\textwidth}
            \centering
            \begin{tabular}{l|ccc:cc|ccc:cc}
Data & \multicolumn{5}{l}{Train.} & \multicolumn{5}{l}{Test} \\ 
Evaluation &    acc. &     $F_1$ & entropy & mae & Var &    acc. &    $F_1$ & entropy & mae & Var \\
\midrule
Baseline      &  0.9790 &  0.0000 &       0.1017 &           0.0212 &       9.25e-04 &  \textcolor{gray}{0.9841} &  0.0 &       0.0823 &           0.0189 &       5.46e-04 \\
Logist. Regr. &  0.9790 &  0.0000 &       0.0903 &           0.0147 &       6.65e-04 &  \textcolor{gray}{0.9841} &  0.0 &       0.2008 &           0.0150 &       8.25e-04 \\
Random Forest &  \textcolor{gray}{0.9791} &  \textcolor{gray}{0.0048} &       \textcolor{gray}{0.0767} &           0.0051 &       1.14e-04 &  \textcolor{gray}{0.9841} &  0.0 &       0.0674 &           0.0042 &       8.86e-05 \\
XGBoost       &  \textcolor{gray}{0.9791} &  0.0004 &       0.0812 &           \textcolor{gray}{0.0030} &       \textcolor{gray}{3.76e-05} &  \textcolor{gray}{0.9841} &  0.0 &       \textcolor{gray}{0.0665} &           \textcolor{gray}{0.0024} &       \textcolor{gray}{3.08e-05} \\
NN - bagging  &  0.9790 &  0.0000 &       0.0849 &           0.0107 &       3.56e-04 &  \textcolor{gray}{0.9841} &  0.0 &       0.0717 &           0.0093 &       3.58e-04 \\
NN - boosting &  \textcolor{gray}{0.9791} &  0.0026 &       0.0840 &           0.0090 &       3.33e-04 &  0.9838 &  0.0 &       0.1104 &           0.0141 &       8.58e-04 \\
\bottomrule
\end{tabular}

            \subcaption{Profile 4}
            \label{table:stats4}
        \end{subtable}
    \caption{Statistics of the discussed classifiers for surrender profiles 1-4. }
    \label{table:stats}
\end{table}

\paragraph{\textcolor{\editcolor}{Numeric evaluation.}} \textcolor{\editcolor}{Next,} we provide \textcolor{\editcolor}{statistics on} latent quantities, as the mean \textcolor{\editcolor}{absolute error} mae$(p,\hat{p})$ (in short: mae) and the variance (Var) of the errors $p(1\vert x)-\hat{p}(1\vert x)$, but also observable quantities as the accuracy (acc.), the $F_1$-score and the binary cross-entropy (entropy). All results are provided for training and test data separately. In each column we highlight the best value. Note that computing the variance of the errors $p(1\vert x)-\hat{p}(1\vert x)$ is fundamentally different from computing the variance of label predictions $\hat{p}(1\vert x)$. While the later is observable and \textcolor{\editcolor}{is minimized by a constant estimator}, it disregards the fact that every contract $x$ induces an individual random variable $(Y\vert X=x)$, which is to be modeled by $\hat{p}(1\vert x)$. Hence, the variance of predictions $\hat{p}(1\vert x)$ holds little value as it penalizes e.g. an unbiased classifier if true probabilities $p(1\vert x)$ evenly populate $[0,1]$.\\
\textcolor{\editcolor}{Table \ref{table:stats} confirms our observations on latent quantities visualized in Figure \ref{fig:pq_plots}. XGBoost shows the lowest mean and variance of absolute errors across all profiles, closely followed by the random forest. In terms of mean and variance, the neural networks outperform the logistic regression. In profile 4, the statistics confirm that the logistic regression and the neural booster both handle the induced stationary noise rather poorly. Overall, the we see significant improvements of the naive baseline model, with its mean absolute error being outperformed by the best cross-entropy model by about $1.5\%$ up $6.2\%$. The alternative models show very similar improvements.}\\
\textcolor{\editcolor}{Looking at the observable quantities, we first note very little differences in the accuracy of the models and $F_1$-score that are approximately $0$, except for profile 1. In their computation we use a default threshold of $0.5$. These observations are expected in a rare event setting and we will apply resampling to improve the $F_1$-score in the upcoming paragraph. For the cross-entropy, it is notable that the random forest throughout performs best on the training data, but does not generalize as well as the XGBoost classifier. Further, the cross-entropy reliably manages to sort the performance by model type, i.e. tree-based, followed by neural networks and logistic regression. However, the ranks of cross-entropy and mae values are not perfectly aligned, although highly correlated, see Table \ref{table:stats}. The final paragraph on time-dependent evaluations provide an additional diagnostic tool that aims to mitigate this discrepancy. First, however, we will try to improve the $F_1$-score, i.e. the balance between recall and specificity, by resampling and monitor distributional affects, in particular in view of Section \ref{section:rare_events} and Lemma \ref{lemma:bias_resampling} and \ref{lemma:skew_resampling}.}

\begin{remark}[ROC]
    \textcolor{\editcolor}{In addition to the statistics in Table \ref{table:stats}, 
    we inspected the ROC and precision-recall curves (PRC) and their corresponding AUC values of all models, evaluated on test data $\mathcal{D}_{\text{test}}$. Overall, the results indicated a very limited benefit of tuning the decision threshold $c\in[0,1]$, providing motivation for alternative approaches such as resampling. The respective figures are provided in the Appendix \ref{appendix:figures}. We find all models to be hardly distinguishable with respect to their ROC, which is insensitivity towards the underlying rare event character. For PRC, we generally observe rather flat curves with low AUC values of no more than $0.12$ for profiles 2,3 and 4, which indicate low $F_1$-scores across all choices of thresholds. For profile 1, we observe much higher AUC values of up to $1$ (ROC) and $0.6$ (PRC). By construction, the range of the true probabilities $p(1\vert x)$ in profile 1 is wider than for other profiles, see Figure \ref{fig:pq_plots}, which makes it easier to separate classes. At the same time, the PRC still fails to identify the significant underestimation of risk of the logistic regression, which we observed e.g. in Figure \ref{fig:pq_1}. Interestingly, for all profiles the AUC values of the PRC rank tree based methods first, followed by neural networks and logistic regression. }
\end{remark}

\paragraph{\textcolor{\editcolor}{The effect of resampling.}}
We retrain all models discussed \textcolor{\editcolor}{in Section \ref{section:models}, now using balanced data generated by} random undersampling or SMOTE resampling. \textcolor{\editcolor}{We employ the same set of hyperparameters as presented in Section \ref{section:models}, as the respective parametrizations have shown the capacity to capture the surrender profiles. Also, hyperparameter tuning for resampled data would be seed dependent, in particular for undersampling, and does not generalize well. For the sake of brevity, we do not present results for all four profiles in the undersampling and SMOTE resampling setting, but select profile 3. The results hereinafter are representative for all other combinations of profiles and resampling schemes and are in line with the theory in Section \ref{section:rare_events}.}
    
    \begin{table}[htb]
        \centering
        \begin{subtable}{\textwidth}
            \centering
            \begin{tabular}{l|ccccc|ccccc}
Data & \multicolumn{5}{l}{Train.} & \multicolumn{5}{l}{Test} \\ 
Evaluation &    acc. &     $F_1$ & entropy & mae & Var &    acc. &    $F_1$ & entropy & mae & Var \\
\midrule
Baseline      &  \textcolor{gray}{0.9799} &  0.0000 &       \textcolor{gray}{0.0984} &           \textcolor{gray}{0.0206} &           \textcolor{gray}{0.0009} &  \textcolor{gray}{0.9838} &  0.0000 &      \textcolor{gray}{ 0.0831} &           \textcolor{gray}{0.0185} &           \textcolor{gray}{0.0005} \\
Logist. Regr. &  0.6646 &  0.0856 &       0.5643 &           0.3564 &           0.0537 &  0.6996 &  0.0740 &       0.5096 &           0.3302 &           0.0527 \\
Random Forest &  0.7315 &  \textcolor{gray}{0.1118} &       0.4941 &           0.3140 &           0.0523 &  0.7619 &  0.0919 &       0.4529 &           0.2935 &           0.0502 \\
XGBoost       &  0.7426 &  0.1100 &       0.4909 &           0.3112 &           0.0503 &  0.7851 &  \textcolor{gray}{0.0991} &       0.4342 &           0.2849 &           0.0465 \\
NN - bagging  &  0.6923 &  0.0956 &       0.5357 &           0.3280 &           0.0631 &  0.7333 &  0.0826 &       0.4783 &           0.3018 &           0.0595 \\
NN - boosting &  0.6637 &  0.0863 &       0.5875 &           0.3664 &           0.0535 &  0.7021 &  0.0740 &       0.5299 &           0.3397 &           0.0522 \\
\bottomrule
\end{tabular}

            \subcaption{Resampling scheme: random undersampling}
            \label{table:undersampling_2}
        \end{subtable}
    \end{table}
    \begin{table}[htb]\ContinuedFloat
            \centering
            \begin{subtable}{\textwidth}
                \centering
                \begin{tabular}{l|ccccc|ccccc}
Data & \multicolumn{5}{l}{Train.} & \multicolumn{5}{l}{Test} \\ 
Evaluation &    acc. &     $F_1$ & entropy & mae & Var &    acc. &    $F_1$ & entropy & mae & Var \\
\midrule
Baseline      &  \textcolor{gray}{0.9799} &  0.0000 &       \textcolor{gray}{0.0984} &           \textcolor{gray}{0.0206} &           \textcolor{gray}{0.0009} &  \textcolor{gray}{0.9838} &  0.0000 &       \textcolor{gray}{0.0831} &           \textcolor{gray}{0.0185} &           \textcolor{gray}{0.0005} \\
Logist. Regr. &  0.6538 &  0.0843 &       0.5541 &           0.3482 &           0.0572 &  0.6883 &  0.0730 &       0.4976 &           0.3206 &           0.0559 \\
Random Forest &  0.7586 &  0.1151 &       0.4498 &           0.2829 &           0.0556 &  0.8065 &  0.1014 &       0.3795 &           0.2451 &           0.0495 \\
XGBoost       &  0.7663 &  0.1137 &       0.4571 &           0.2837 &           0.0545 &  0.8260 &  \textcolor{gray}{0.1064} &       0.3722 &           0.2428 &           0.0455 \\
NN - bagging  &  0.7523 &  \textcolor{gray}{0.1173} &       0.4523 &           0.2698 &           0.0684 &  0.7846 &  0.0908 &       0.4052 &           0.2417 &           0.0640 \\
NN - boosting &  0.7210 &  0.1024 &       0.4884 &           0.2977 &           0.0641 &  0.7535 &  0.0872 &       0.4359 &           0.2697 &           0.0616 \\
\bottomrule
\end{tabular}

                \subcaption{Resampling scheme: SMOTE}
                \label{table:SMOTE_2}
            \end{subtable}
        \caption{Statistics for profile 2 and models calibrated on resampled data. }
        \label{table:stats_resampling}
    \end{table}

    \begin{figure}[htb]
        \centering
        \begin{subfigure}{.9\textwidth} 
            \centering
            \includegraphics[width=\textwidth]{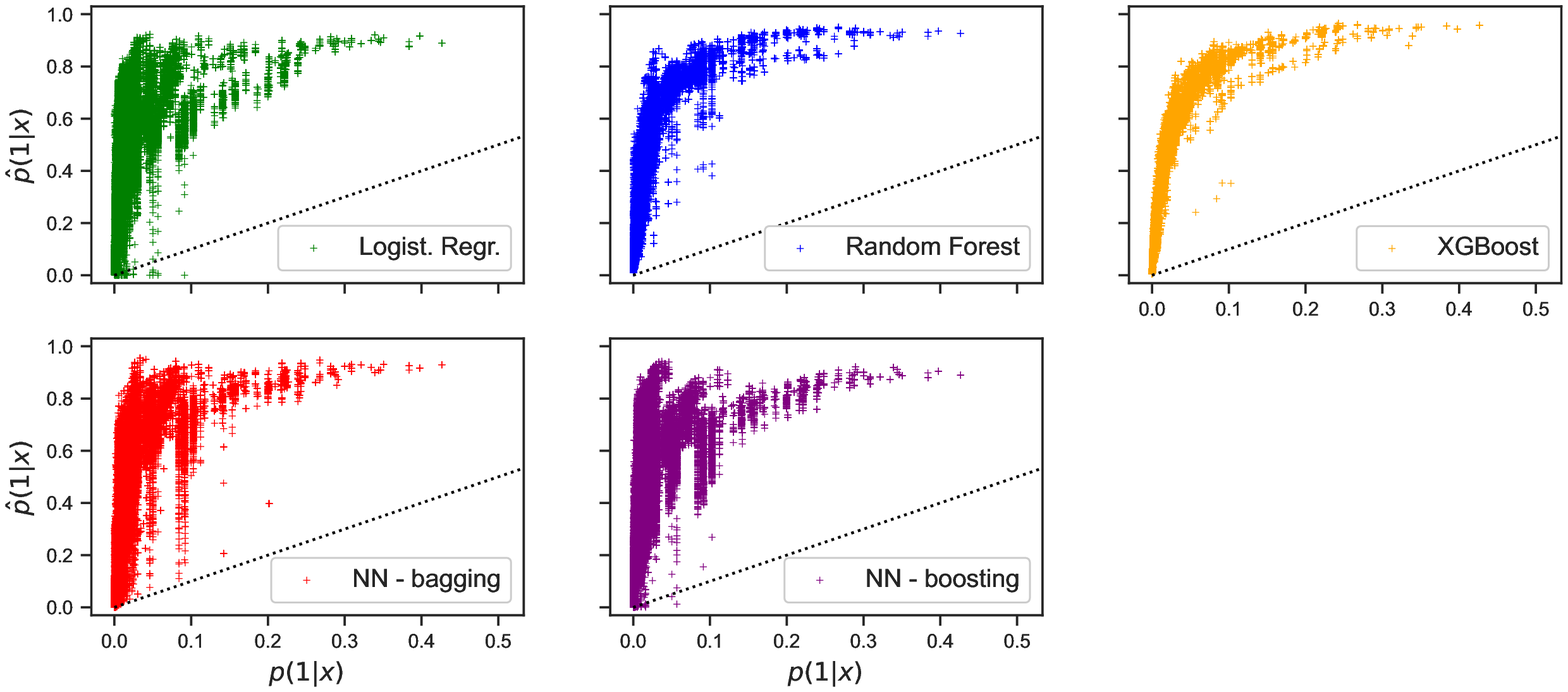}
            \subcaption{Resampling scheme: random undersampling.}
            \label{fig:undersampling_2}
        \end{subfigure}
    \end{figure}
    \begin{figure}[htb]\ContinuedFloat
        \centering
        \begin{subfigure}{.9\textwidth} 
            \centering
            \includegraphics[width=\textwidth]{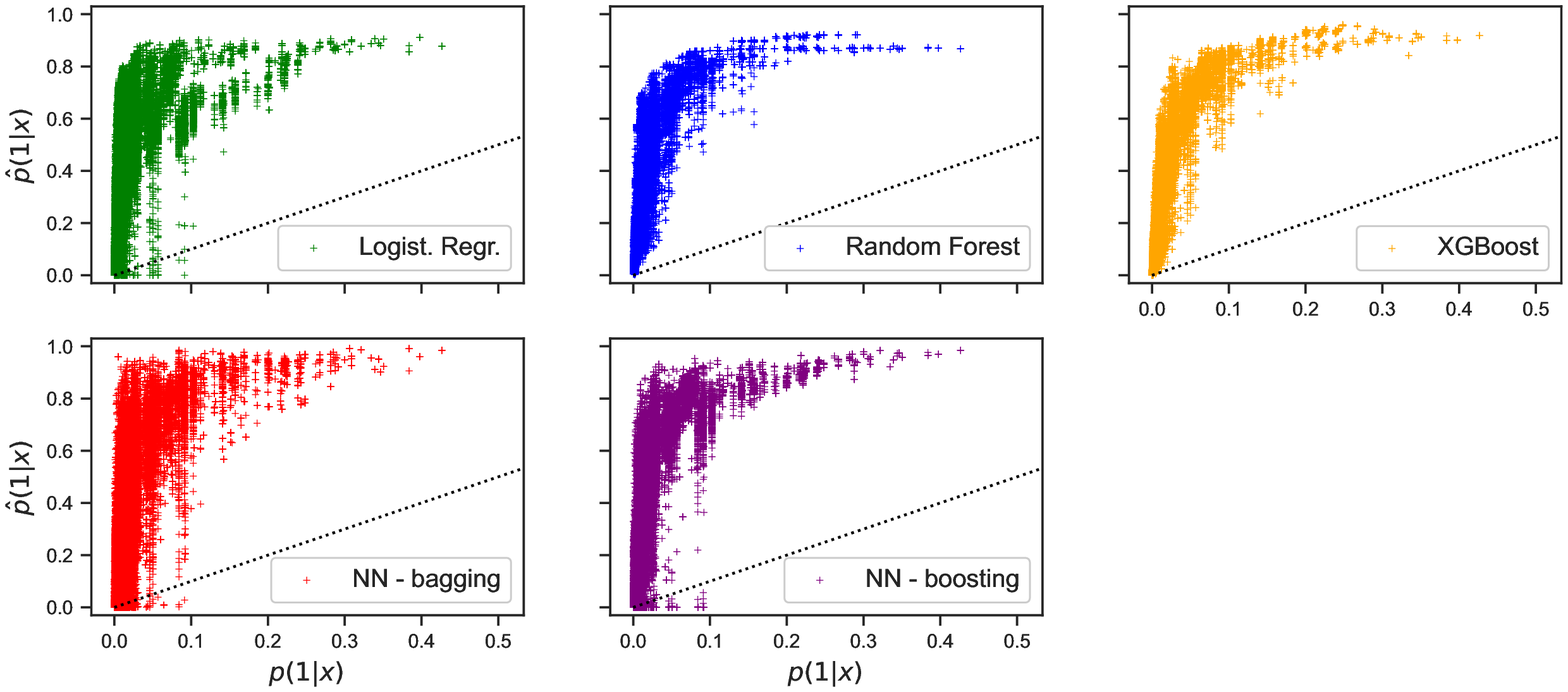}
            \subcaption{Resampling scheme: SMOTE.}
            \label{fig:SMOTE_2}
        \end{subfigure}
        \caption{$p$-$\hat{p}^S$ plots for surrender profile 2 for data $x\in\mathcal{D}_{\text{test}}$ and models formed on resampled data.}
        \label{fig:pq_resampling}
    \end{figure}

In Table \ref{table:stats_resampling} we revisit the evaluation of our classifiers, which now have been trained on resampled data. \textcolor{\editcolor}{Using a threshold $c=0.5$ to form label predictions}, we notice that \textcolor{\editcolor}{SMOTE} resampling indeed improves the $F_1$-score of all classifiers on the training as well as the test data, \textcolor{\editcolor}{compare with Table \ref{table:stats3}}. This means that \textcolor{\editcolor}{both resampling schemes} positively affect the balance between correct label predictions $\hat{y}=1$ and correctly classified observations $y=1$. While this might be the objective in situations where binary actions have to be taken, e.g. the decision to investigate for fraud, Table \ref{table:stats_resampling} also shows a downside \textcolor{\editcolor}{of resampling}. First, there is a trade-off between $F_1$-score and accuracy. Secondly and arguably more importantly, we observe that the cross-entropy as well as the mean and variance of the latent errors $\vert p(1\vert x)-\hat{p}^{S}(1\vert x)\vert$ have increased for all classifiers on training and test data, compared to training without resampling, see Figure \ref{table:stats3}. \\ 
\textcolor{\editcolor}{Lemma \ref{lemma:bias_resampling} and \ref{lemma:skew_resampling} in Section \ref{section:rare_events} give information about why the models do not capture the true probability of surrender properly. Looking at $p$-$\hat{p}^{S}$-plots in Figure \ref{fig:pq_resampling}, we can numerically affirm both lemmas and observe the overestimation of the surrender risk, as well as the concave shape of that bias due to resampling.} Note that both lemmas describe the relation between $\hat{p}$ and $\hat{p}^S$, \textcolor{\editcolor}{ whereas we depict the $\hat{p}$ and $p$, since the later is the latent quantity of interest.} \textcolor{\editcolor}{If we compare the logistic regression models in Figures \ref{fig:pg_3} and \ref{fig:pq_resampling}, we even observe that the bias introduced by resampling dominates in a way that models previously underestimating the true probability $p$, after resampling, now clearly overestimate the quantity. \\
For completeness, we also investigated receiver operating curves and precision recall curves and their respective AUC values of models trained on resampled training data $\mathcal{D}_{\text{train}}$ and evaluated on test data $\mathcal{D}_{\text{test}}$. However, we found neither global nor local improvements in the respective curves. In fact, the curves closely resemble the curves prior to resampling. This indicates that in our setting, where surrender is a rare and low-probability event, models trained on resampled data did not find a more selective separation of the surrender and non-surrender classes. Illustrations of the curves can be found in the Appendix \ref{appendix:figures}, Figure \ref{fig:roc_resampling}.} \\
\textcolor{\editcolor}{Altogether, the provided results strongly advise against the use of resampling if we are interested in an unbiased modeling of the true probability of surrender. Optimizing models based on frequentist performance measures, our main motivation for and a common benefit of applying resampling, turns out to be in conflict with our probabilistic objective. Low $F_1$-scores or low AUC values of precision-recall curves show to be natural consequence of the rare event setting with hardly separable classes and do not indicate poor model performance. In the next paragraph, we provide a final, probabilistic evaluation that complements the information of the cross-entropy loss. }

\paragraph{\textcolor{\editcolor}{Time-dependent evaluation.}}
At last, we omit resampling and evaluate all classifiers in a practical setting, where we cannot access the latent probabilities $p(1\vert x)$. We now take a time-series perspective and look at predicted mean surrender rates and their confidence bands over calendar-time, as proposed in Corollary \ref{cor:confidence_intervals}. \textcolor{\editcolor}{These are then compared to the true, observed surrender rates $\tfrac{1}{\vert Y\vert}\sum_{i\subset Y} y_i$, where $Y\subset \lbrace1, 2, \ldots, N\rbrace$ is an index set of contracts from a specific calendar year. } This concept uncovers the quality of the proposed classifiers in a going concern setting and respects that an insurance portfolio \textcolor{\editcolor}{or the behaviour of policyholders} might change over time, leading to non-constant risk exposure e.g. in term of a value-at-risk measure. \textcolor{\editcolor}{Overall, we will observe the superior performance of XGBoost and random forest, but now in a more detailed and informative way than examining a single cross-entropy value, as in Table \ref{table:stats}.} The \textcolor{\editcolor}{new evaluation is} illustrated in Figure \ref{fig:msr_plots}, where the split in calendar time between train and test data is indicated by a vertical, gray line. Upper and lower $\alpha=0.95$ confidence bands, as in Corollary \ref{cor:confidence_intervals}, are displayed by up- respectively down-facing triangles.

\begin{figure}[htb] 
        \centering
        \begin{subfigure}{.49\textwidth} 
            \centering
            \includegraphics[width=.95\textwidth]{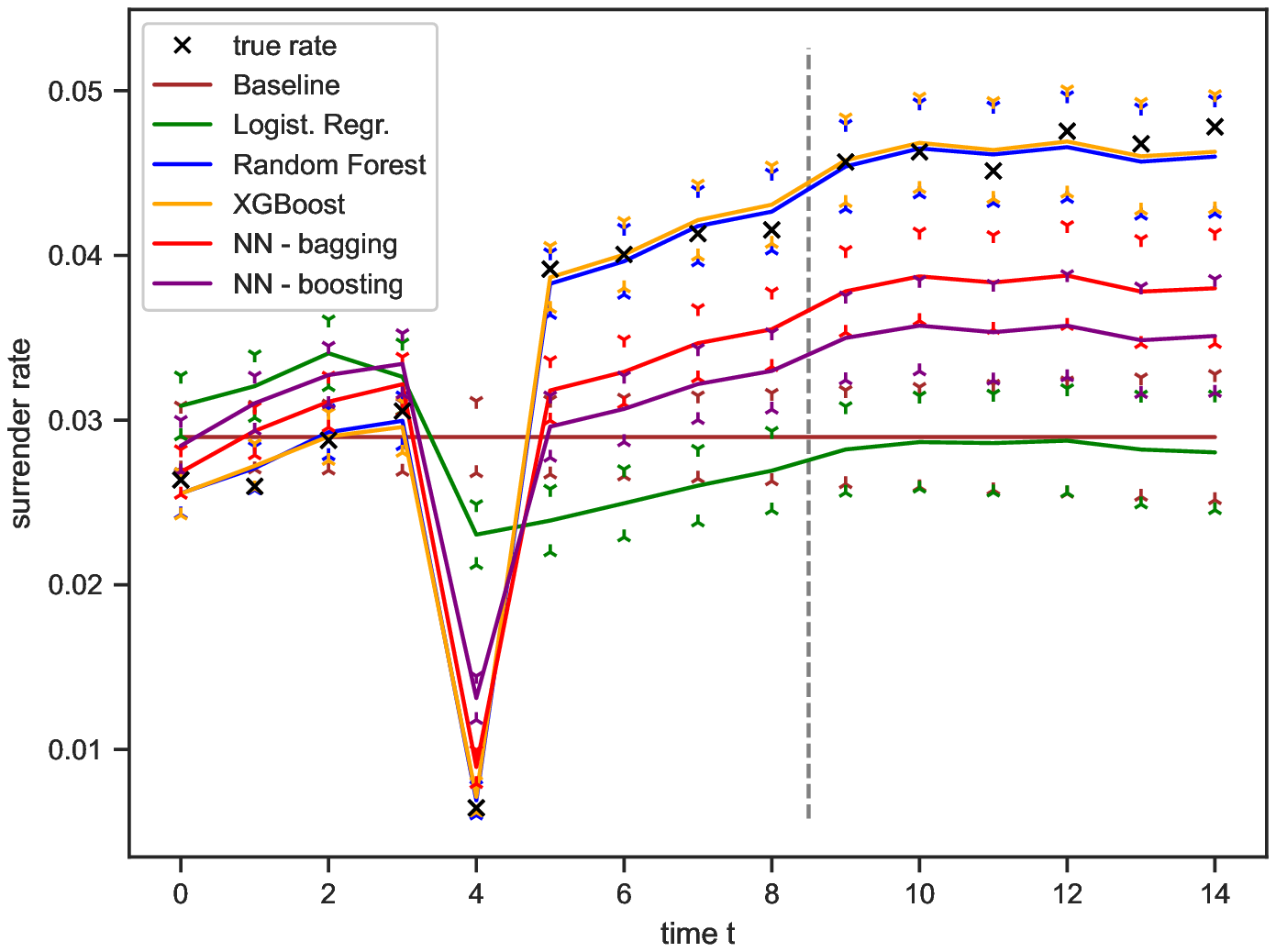}
            \subcaption{Profile 1}
            \label{fig:msr_1}
        \end{subfigure}
        \begin{subfigure}{.49\textwidth} 
            \centering
            \includegraphics[width=.95\textwidth]{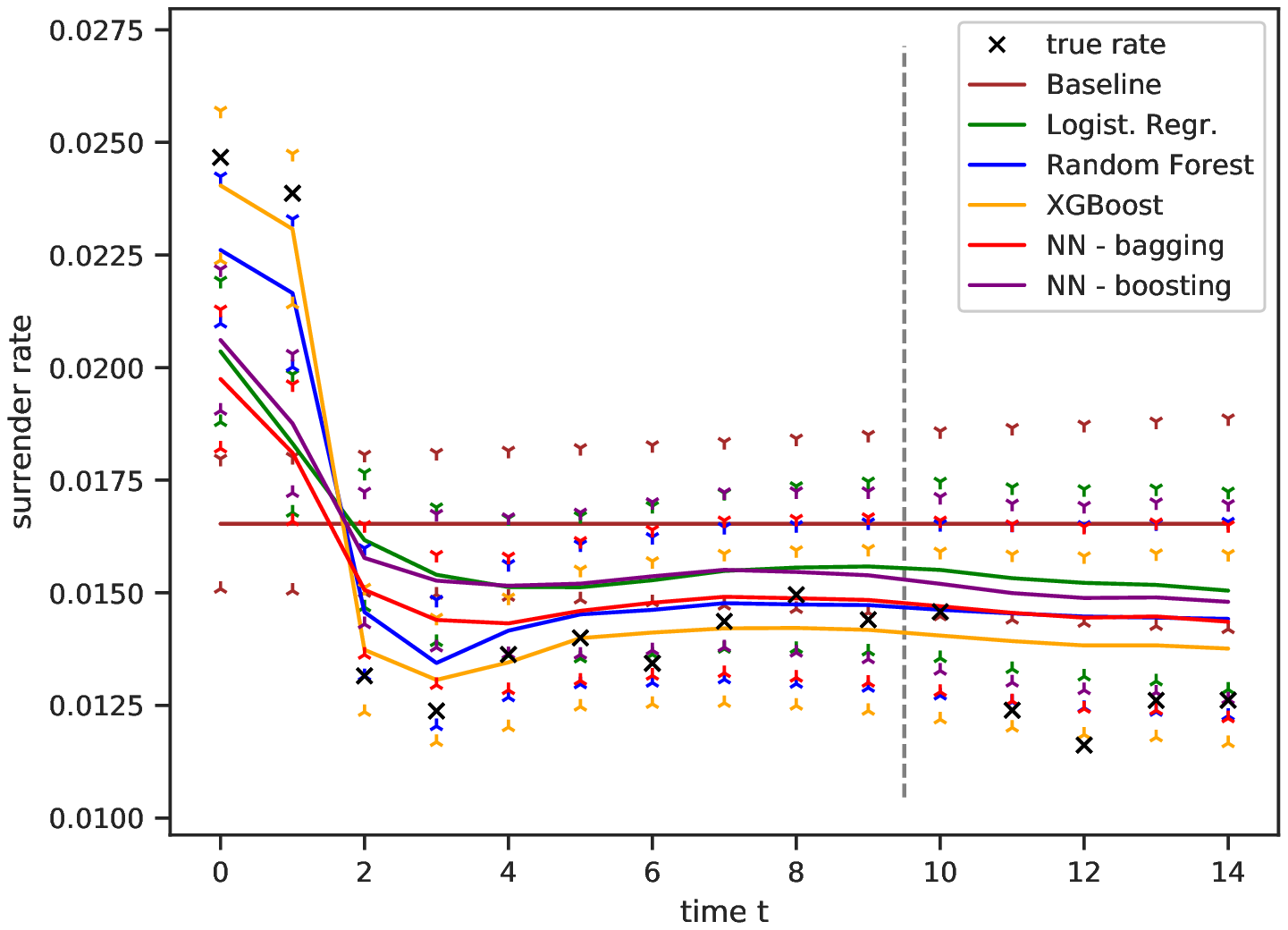}
            \subcaption{Profile 2}
            \label{fig:msr_2}
        \end{subfigure}
        \begin{subfigure}{.49\textwidth} 
            \centering
            \includegraphics[width=.95\textwidth]{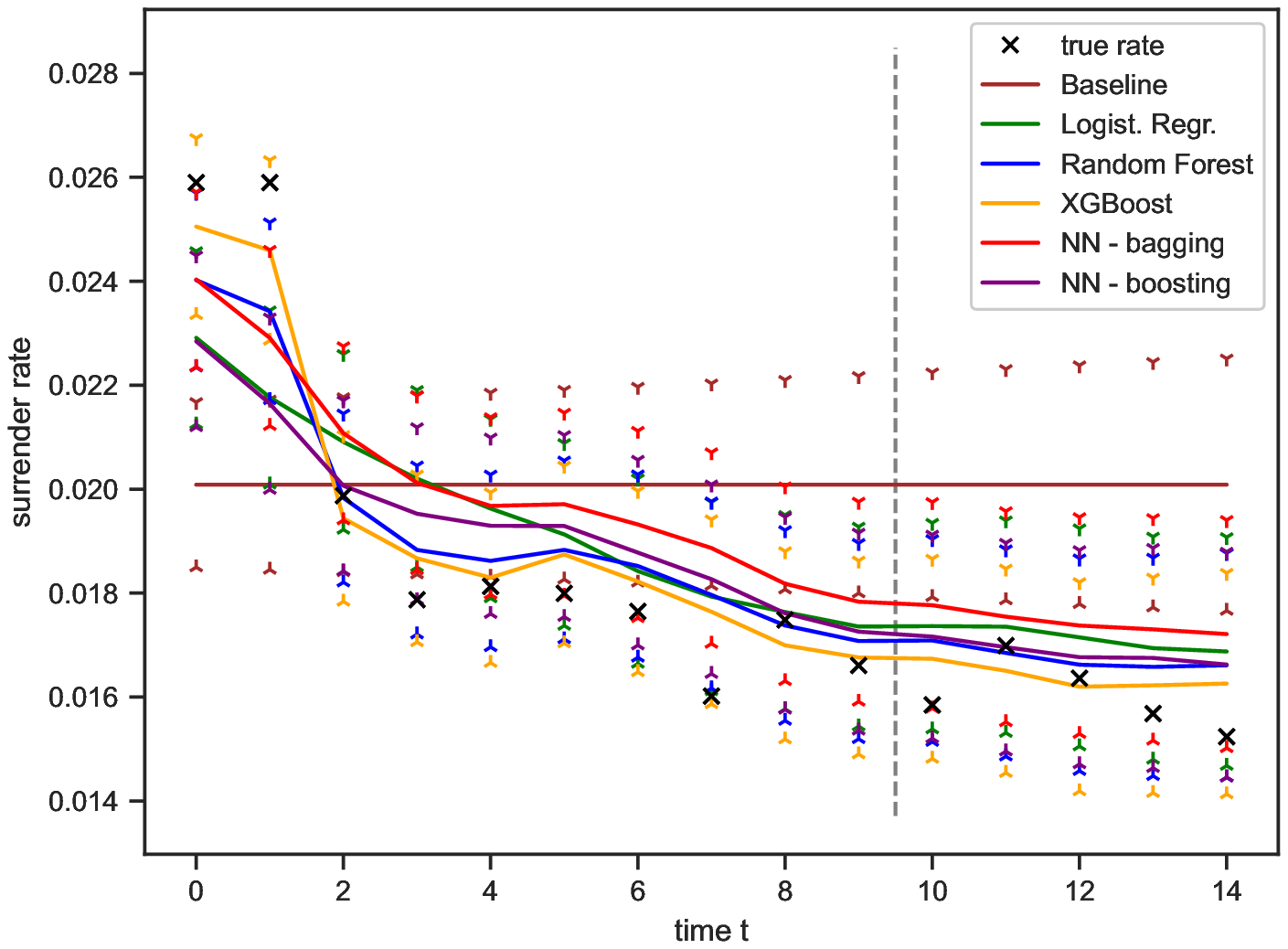}
            \subcaption{Profile 3}
            \label{fig:msr_3}
        \end{subfigure}
        \begin{subfigure}{.49\textwidth} 
            \centering
            \includegraphics[width=.95\textwidth]{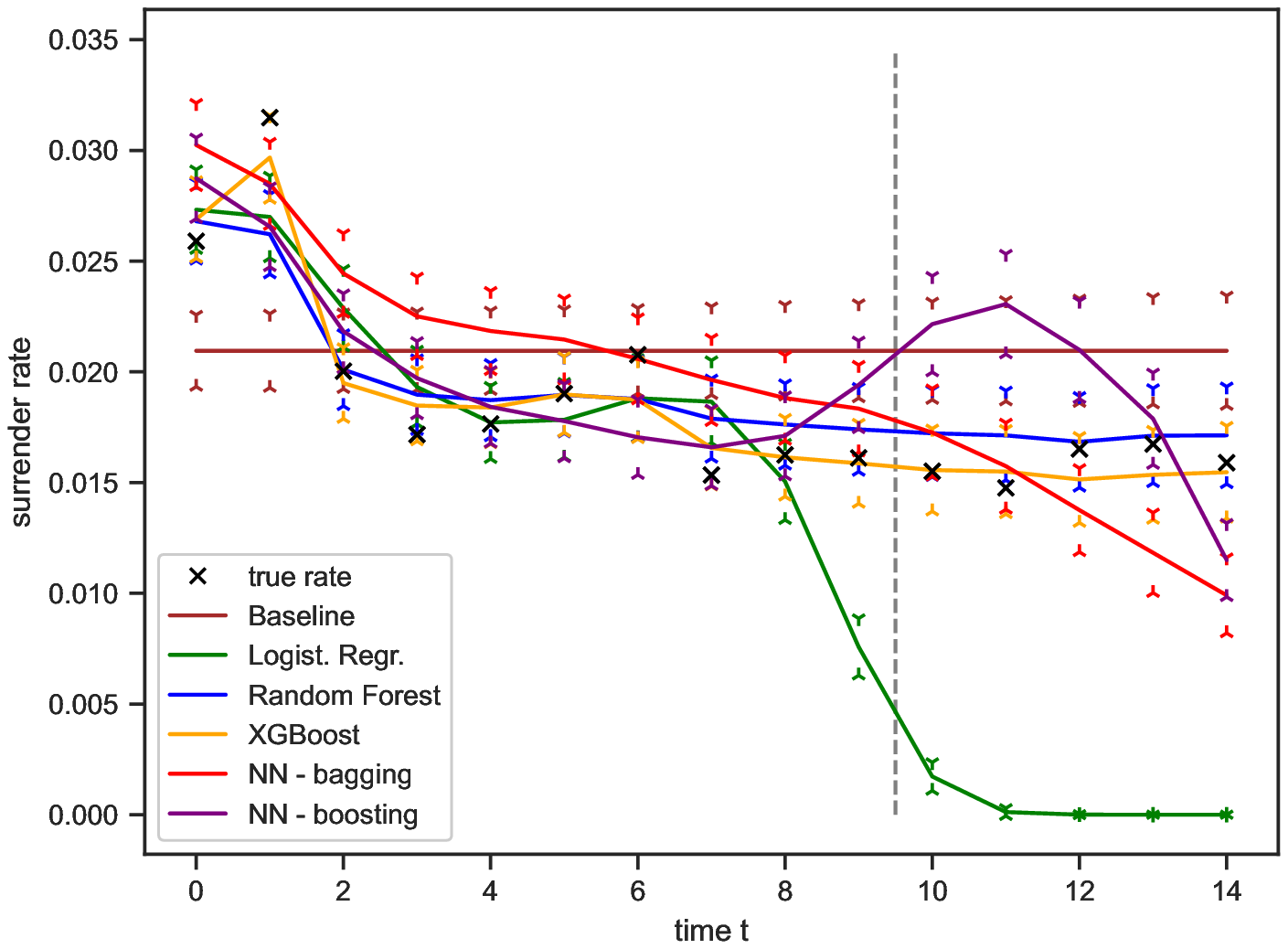}
            \subcaption{Profile 4}
            \label{fig:msr_4}
        \end{subfigure}
    \caption{Mean surrender rates over calendar-time of the classifiers for surrender profiles 1-4, including  confidence bands ($\alpha=0.95$) based on Corollary \ref{cor:confidence_intervals}. Vertical lines indicate a train-test-split.}
    \label{fig:msr_plots}
\end{figure}

First, we observe that the mean surrender rates are non-stationary. Profiles 2-4 show mean surrender rates that decline over time. In contrast, in profile 1 we see an increase of mean surrender activity with a notable drop at calendar year 4. Recall that surrender profile 1 as in \cite{Milhaud.2011} exhibits a significant increase in surrender activity just prior to an elapsed duration of 4 years, followed by a sharp decline. \\
For all surrender profiles, the naive baseline provides an insufficient classifier, which naturally does not capture any of the non-stationarity. Hence, a constant risk buffer does not appropriately compensate for the simplicity of the model. In profile 1, we observe that all classifiers except the logistic regression capture the trend of mean surrender, including the characteristic, sharp drop at calendar year 4. However, XGBoost \textcolor{\editcolor}{and random forest} are the only classifiers for which the true rate lies consistently within the confidence band of its predictions. \textcolor{\editcolor}{Further, neural networks show to be adaptive as the composition of the portfolio changes, but underestimate the rate from year 5 onward. Recall that the logistic regression can perfectly replicate all underlying surrender profiles, given perfect preprocessing. Hence, Figure \ref{fig:msr_1} also indicates the advantages of implicit data preprocessing by hierarchical models. This seems particularly desirable in a rare event setting where for example too few surrender events are available to accurately determine highly relevant groups of ages. } In profiles 2 and 3, for most years the quality of the proposed classifiers, excluding the naive baseline, cannot be distinguished based on their $0.95$-confidence bands. \textcolor{\editcolor}{However, looking at the best estimate, i.e. the solid lines in Figure \ref{fig:msr_plots}, the tree-based methods predict accurate surrender rates more confidently, with XGBoost outperforming the random forest model.} 
\textcolor{\editcolor}{Lastly, we add non-stationary noise to profile 4, thereby creating profile 4. In this challenging setting, we see a overall decent fit on the training data, i.e. left to the dotted gray line. As we move into the test set, we observe all but the tree-based models confidently fail to capture the true rate of surrender and extrapolate the non-stationary noise. The logistic regression even drifts away from the true, observed rate at the last calendar year in the training data and overall generalizes poorly. As we employed equal diligence for all surrender profiles during hyperparameter tuning, including a $3$-fold cross-validation, Figure \ref{fig:msr_4} shows that the given logistic regression is incapable of capturing the dynamics of surrender, despite being closely related to the meta model for surrender. Further, the bagged neural network gradually drifts away and underestimates the risk, while the boosted neural network highly fluctuates around the true quantity.} In that sense, results on profile 4 also indicate tree based classifiers to provide more stable estimates when characteristics of training and test set differ\textcolor{\editcolor}{, which is very plausible, given their prior assumption on locally constant rates. It is interesting to note that we can use the presented confidence bands not only to identify well-performing models, but also to monitor our assumption and potentially identify changing behaviour of policyholders. In profile 4, the noise can be viewed as the economic environment or arbitrary information, which we do not capture explicitly. As the trend of the noise changes, recall Figure \ref{fig:surrender_profile_4}, the mean surrender rates and their confidence bands reflect that change of the underlying conditions.}
\section{Conclusion} \label{section:conclusion}
In the present work, we perform extensive numerical experiments, where we look at four different surrender profiles, each motivated by empirical research in the literature.  Our modeling approaches include the logistic regression, tree based classifiers and neural networks, each in a bagged and a boosted version. The results generally indicate a highly practical performance, where XGBoost is the superior model across all surrender profiles, \textcolor{\editcolor}{closely followed by the random forest and eventually neural networks}. During the evaluation we \textcolor{\editcolor}{first look at latent quantities, such as the bias and the variance, given the true surrender probability of each contract and then aim to extract latent information from observable quantities.} With regard to low bias, our results indicate a trade-off between accurate label prediction for rare events, such as surrender, and a sound estimation of the true surrender probabilities. We also numerically affirm earlier findings in the literature that common, \textcolor{\editcolor}{consistent} resampling techniques can indeed improve frequentist evaluations of our classifiers, as e.g. the $F_1$-score. However, \textcolor{\editcolor}{in our setting these improvements do not hold across all potential choices of decision thresholds $c\in[0,1]$. Further,} our theoretical and numerical results show that \textcolor{\editcolor}{resampling} comes along with a significant bias of the predicted surrender probabilities. Hence, resampling severely skews the underlying risk and is impractical e.g. with regard to a value-at-risk assessment of surrender risk. \textcolor{\editcolor}{Overall, our analysis strongly advises against the use of resampling, when the objective is a sound prediction of the surrender probability.} \\
Lastly, \textcolor{\editcolor}{we observe the cross-entropy to be a good proxy of the latent bias, at least in terms of ranking models by their performance.} To \textcolor{\editcolor}{disentangle this single value and to} promote sound model evaluation, \textcolor{\editcolor}{we propose confidence bands on mean surrender rates with respect to calendar-time as a complementary evaluation.} This allow us to identify poor performance in a going concern perspective, where the composition of a portfolio or the predominant risk drivers might change over time. \textcolor{\editcolor}{At the same time, we observe that confidence bands can be used as a diagnostic tool to identify when the behaviour of policyholders changes. Overall, this final evaluation provides not only the correct ranking of model performance, but also more insight into at what time models start to incorrectly predict surrender. In particular}, our experiments highlight that adding a risk buffer to a naive baseline does not cover the surrender risk sufficiently. This observation indicates the importance of time for a practical model evaluation and provides support for a probabilistic, time-series perspective on surrender risk. \\

Further research should focus on the application of our findings to high-dimensional, real data sets. We are interested in the benefit of confidence bands for mean surrender rates on data sets, where prior analysis was based on frequentist concepts and the application of resampling. Confidence bands could also be used to monitor if a given model needs to be recalibrated as either the composition of the portfolio or the behaviour of policyholders change. Further, it would be interesting to extend our results by including the loss-given-surrender, or an appropriate deterministic proxy, and thereby evaluate models based on their estimate of the economic capital at risk due to surrender. The presented concept of mean event probabilities and confidence bands is not restricted to surrender or one year time horizons, but is valid for general Bernoulli-type events. Hence, it can also be used to monitor e.g. the quality of a model for expected fraud rates or one-year death events.

\section*{Acknowledgements}
The author wants to thank his supervisor Christian Weiß for his precious feedback and numerous helpful discussions. 

\printbibliography

\appendix
\section{Appendix}

\subsection{Proofs} \label{appendix:proofs}

    \begin{proof}[Proof of Proposition \ref{prop:CLT}] \label{proof:prop:CLT}
        In order to apply the central limit theorem of Lindeberg-Feller, it remains to show that the Lindeberg condition holds, see Theorem 15.43 in \cite{Klenke.2014}. Hence, for all $\delta>0$ we have to show that
        \begin{align} \label{eq:clt_proof}
            L_N(\delta):&=\frac{1}{\sigma^2(S_N)}\sum_{i=1}^{N} \mathbb{E}\left[ (Z_i-p(1\vert x_i))^2\mathds{1}_{\lbrace \vert Z_i-p(1\vert x_i)\vert>\delta \sigma(S_N) \rbrace}\right] \xrightarrow[N \to \infty]{} 0 .
        \end{align}
        Observe that $\vert Z_i-p(1\vert x_i)\vert$ is bounded, while $\sigma(S_N)^2=\sum_{i=1}^{N}p(1\vert x_i)(1-p(1\vert x_i))\xrightarrow[N \to \infty]{} \infty$, given that $p(1\vert x_i)\in[\varepsilon, 1-\varepsilon]$ for all $i\in\mathbb{N}$. Hence, for every $\delta>0$ there is $\Bar{N}_{\delta}\in \mathbb{N}$ such that $\vert Z_i-p_{x_i} \vert<\delta\sigma(S_N),$ for all $N>\Bar{N}_{\delta}$ and for all $i\in\mathbb{N}$. Therefore, the sum in \eqref{eq:clt_proof} contains at most $\Bar{N}_{\delta}$ bounded summands and we conclude $L_N(\delta)\xrightarrow[N \to \infty]{} 0$.
    \end{proof}

    \begin{proof}[Proof of the Corollary \ref{cor:confidence_intervals}] \label{proof:cor:CLT}
        By Lemma \ref{prop:CLT}, we know 
        \begin{align}
            \frac{1}{\sigma(S_N)}\sum_{i} Z_i - p(1\vert x_i) \overset{d}{\rightarrow} \mathcal{N}\left(0,1\right).
        \end{align}
        Then, for large values of $N$
        \begin{align}
            1-\frac{\alpha}{2} & = \mathbb{P}\left(\tfrac{1}{N}\sum_i Z_i \leq z \right) \approx \Phi\left(\frac{N~ z-\sum_i p_{x_i}}{\sqrt{\sum_i p(1\vert x_i)(1-p(1\vert x_i)}}\right),
                    \intertext{holds, and therefore}
            z & \approx \frac{1}{N}\sum_i p(1\vert x_i) + \Phi^{-1}\left(\frac{\alpha}{2}\right)~\frac{\sqrt{\sum_i p(1\vert x_i)(1-p(1\vert x_i)})}{N}.
        \end{align}
        
        Let $\hat{p}:x\mapsto \hat{p}(1\vert x)$ describe the estimator for the true Bernoulli probabilities $p:x\mapsto p(1\vert x)$. Then, substituting $\frac{1}{N}\sum_i p(1\vert x_i)$ by its estimate $\frac{1}{N}\sum_i \hat{p}(1\vert x_i)$ and the standard deviation $\sigma\left( S_N\right)=\sqrt{\sum_i p(1\vert x_i)(1-p(1\vert x_i))}$ by its asymptotically unbiased sample estimate $\hat{\sigma}(S_N)=\sqrt{\frac{N}{N-1}}~\sqrt{\sum_i \hat{p}(1\vert x_i)(1-\hat{p}(1\vert x_i))}$ yields the statement.
    \end{proof}
    
    \begin{proof}[Proof of Lemma \ref{lemma:bias_resampling}] \label{proof:lemma:res}
        Using the decomposition proposed in Theorem \ref{theorem:effect_of_resampling} and $\hat{p}^S(y=1)>\hat{p}(y=1)$ yields
        \begin{align*}
            \hat{p}^S(1\vert x) &= \hat{p}(1\vert x)~\frac{\hat{p}^S(y=1)(1-\hat{p}(y=1))}{\hat{p}(y=1)(1-\hat{p}(1\vert x))+\hat{p}^S(y=1)(\hat{p}(1\vert x)-\hat{p}(y=1))} \\
                    & > \hat{p}(1\vert x)~\frac{\hat{p}^S(y=1)(1-\hat{p}(y=1))}{\hat{p}^S(y=1)(1-\hat{p}(1\vert x)+\hat{p}(1\vert x)-\hat{p}(y=1))} \\
                    & = \hat{p}(1\vert x)~.
        \end{align*}
    \end{proof}

\subsection{Simulation of endowment policies.} \label{appendix:simulation_policies}
We present the simulation scheme for the portfolio at the initial calendar year. New business at subsequent calendar years is simulated analogously with the condition that the currently elapsed duration equals zero. Let us denote an arbitrary contract $X=(X^{(1)},\ldots,X^{(n)})$ with $n=8$. The features correspond to the calendar year ($X^{(1)}$), the current age of the policyholder ($X^{(2)}$), the face amount of the contract ($X^{(3)}$),  the duration ($X^{(4)}$), the elapsed duration ($X^{(5)}$), the remaining duration ($X^{(6)}$), the frequency of premium payments ($X^{(7)}$) and the annual amount of premium payments ($X^{(8)}$). For the simulation we use a Pearson gamma distribution $\mathbf{\Gamma}_{\alpha,\lambda}$ with shape $\alpha>0$ and rate $\lambda>0$, see \cite{Borovkov.2013}, as well as the uniform distribution $\mathcal{U}(0,1)$. The gamma distribution provides a flexible, right skewed distribution. For a random variable $\zeta\sim \mathbf{\Gamma}_{\alpha,\lambda}$ the expectation and variance of $\zeta$ are given by
\begin{align*}
    \mathbb{E}(\zeta) = \frac{\lambda}{\alpha}, \quad\quad  \text{Var}(\zeta) = \frac{\lambda}{\alpha^2}.
\end{align*}
In \cite{Milhaud.2018}, we find $34\%$ of policyholders to be of ages 0-34, respectively $34.04\%$ and $18.5\%$ of ages 35-54 and 55-84. This indicates a right-skewed distribution for a policyholder's age. Further, most lapse events, including maturity, are indicated to occur within 15 years. The authors in \cite{Milhaud.2018} also report roughly $25\%$ annual premium payments,  $60\%$ infra annual and $15\%$ supra annual payments. The annual premium amount is reported to be highly right-skewed. We calibrate the marginal distributions of features $X^{(i)}$, $i>1$, to closely imitate these statistics. We define 
\begin{align*}
    X^{(2)} &\sim\Gamma_{5.5,~6.8^{-1}}, & X^{(4)} &\sim 5 + \Gamma_{5,~1.5^{-1}},\\
    X^{(3)} &\sim 5'000 + \Gamma_{4,~2'000^{-1}}, & X^{(5)} &\sim \min(X_4\cdot \mathcal{U}(0,1),X_2).
\end{align*}
As we do not permit negative ages at the underwriting of the contract, the elapsed duration $X^{(5)}$ is engineered to not exceed the current age. The remaining duration $X^{(6)}$ of a contract is then obtained deterministically by $X^{(6)}:= X^{(4)}-X^{(5)}$. Further, the premium frequency $X^{(7)}$ is a categorical variable which mirrors \cite{Milhaud.2018} with
\begin{align*}
    X^{(7)}&=\begin{cases}\text{upfront} & \text{, with prob. } 0.15 \\ \text{annual} & \text{, with prob. } 0.25 \\ \text{monthly} & \text{, with prob. } 0.6 \end{cases}\quad. 
\end{align*}
Last, we determine the fair premium by the equivalence principle in \cite{dickson_hardy_waters_2009, Fuhrer.2006} annualize it linearly based on the premium frequency and the number of remaining premium payments up to maturity, respectively at most up to the last premium payment at retirement, i.e. at age 67. The resulting annual premium $X^{(8)}$ is then used to calibrate the parameters of the face amount $X^{(3)}$. The reported parametrization of $X^{(3)}$ leads to an annual, mean premium of $1'300$ EUR, which inherits the right-skew from $X^{(3)}$. \\

\subsection{Surrender profiles} \label{appendix:surrender_profiles}

All surrender profiles used for the experiments in this work are based on findings in the literature, see \cite{Milhaud.2011, Cerchiara.2008, Eling.2014}, which all employ logistic regressions and report quantitative results on $\beta_{x}^{(i)} x^{(i)}$ for surrender. Empirical results in \cite{Milhaud.2011} are obtained on endowment insurance contracts in Spain, while \cite{Cerchiara.2008} and \cite{Eling.2013} consider life insurance contracts of multiple types in the Italian and German market. Given surrender behaviour on various markets, we do not aim to combine them to a holistic surrender model. Instead we create multiple surrender profiles that capture a variety of the observed risk characteristics. To plausibly combine multiple risk drivers $\beta_{x}^{(i)} x^{(i)}$ we adjust the baseline risk $\beta_0$ such that the observed mean surrender rates of each risk profile fall into practical ranges of 0.01 to 0.05.

\paragraph{Profiles 1 and 2.}
In \cite{Milhaud.2011} the authors use odd ratios to identify quality and severity of risk drivers for endowment insurance products. To compute the odd ratio of a single feature $i$, we take for two contracts $x$ and $\Tilde{x}$, which only differ in feature $i$, i.e. $x^{(j)}=\Tilde{x}^{(j)}$ for $j\neq i$. The odd ratio is then computed by
\begin{align}
    \frac{p(1\vert x)/\left(1-p(1\vert x)\right)}{p(1\vert \Tilde{x})/\left(1-p(1\vert \Tilde{x})\right)} = \exp\lbrace \beta_{x}^{(i)} x^{(i)}-\beta_{\Tilde{x}}^{(i)}\Tilde{x}^{(i)}) \rbrace.
\end{align}
If we set $\beta_{\Tilde{x}}^{(i)}\Tilde{x}^{(i)}:=0$ as the baseline risk of feature $i$, we can extract $\beta_{x}^{(i)} x^{(i)}$ for all features $i$. Consequently, we construct surrender profile 1 based on empirical odd ratios and surrender profile 2 based on modeled odd ratios, both are reported in \cite{Milhaud.2011}. The risk drivers are restricted to the current age of the policyholder, the elapsed duration of the contract, the frequency of premium payments, i.e. monthly, annually or as a lump sum, and the annualized premium amount. Due to data confidentiality in \cite{Milhaud.2011}, the actual premium amounts are omitted. Hence, we transfer the odd ratio for the savings premium in \cite{Milhaud.2011} to our meta model by assuming the indicated jumps to occur at plausible levels of $1'000$ EUR resp. $2'000$ EUR.

\paragraph{Profiles 3 and 4.}
We keep the the effect of the elapsed duration and the annualized premium amount fixed to the values in surrender profile 2. For the premium frequency, as well as the remaining duration we employ results in \cite{Eling.2014}, including regression coefficient estimated on a per contract basis. These results indicate a reduced risk of surrender for single premium payments and an increased risk for a high remaining duration. Further, in \cite{Cerchiara.2008} we find a mitigating effect of mid-range ages, as well as an increased surrender activity for contracts with a duration of less than $10$ years, both reported in the form of $\beta_{x}^{(i)} x^{(i)}$. This completes surrender profile 3. Lastly, for surrender profile 4 we additionally include the calendar year as a risk driver and impose coefficients reported in \cite{Cerchiara.2008} for the years $1991$ up to $2007$. The calendar year can be viewed as a proxy for the economic environment or alternatively as noise to the surrender activity, as input variables $x$ in our data do not include economic features.

\subsection{Figures} \label{appendix:figures}
    
    \begin{figure}[H]
        \centering
        \begin{subfigure}{.95\textwidth} 
            \centering
            \includegraphics[width=\textwidth]{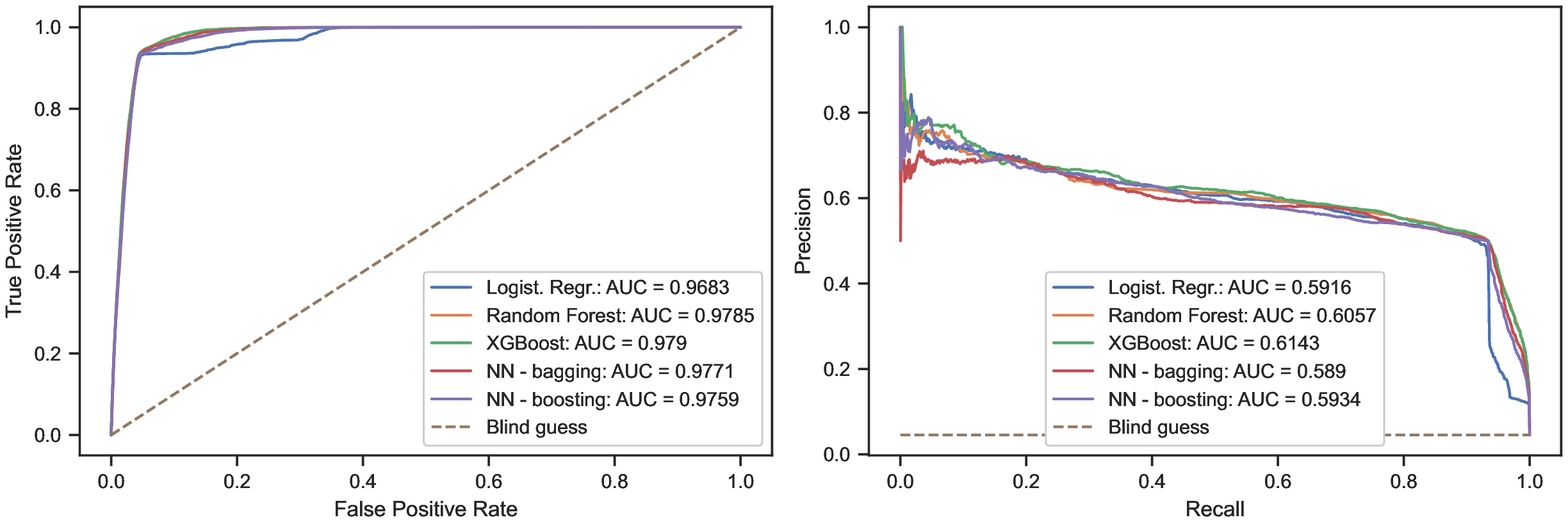}
            \subcaption{Profile 1}
            \label{fig:roc_1}
        \end{subfigure}
        \begin{subfigure}{.95\textwidth} 
            \centering
            \includegraphics[width=\textwidth]{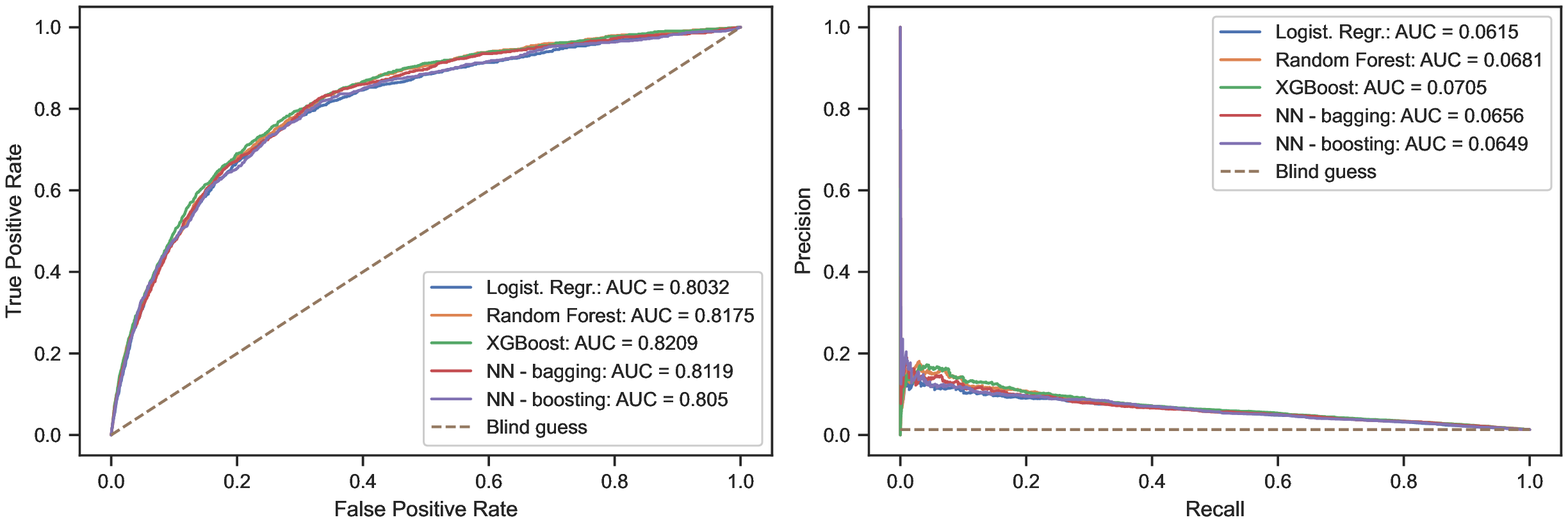}
            \subcaption{Profile 2}
            \label{fig:roc_2}
        \end{subfigure}
        \begin{subfigure}{.95\textwidth} 
            \centering
            \includegraphics[width=\textwidth]{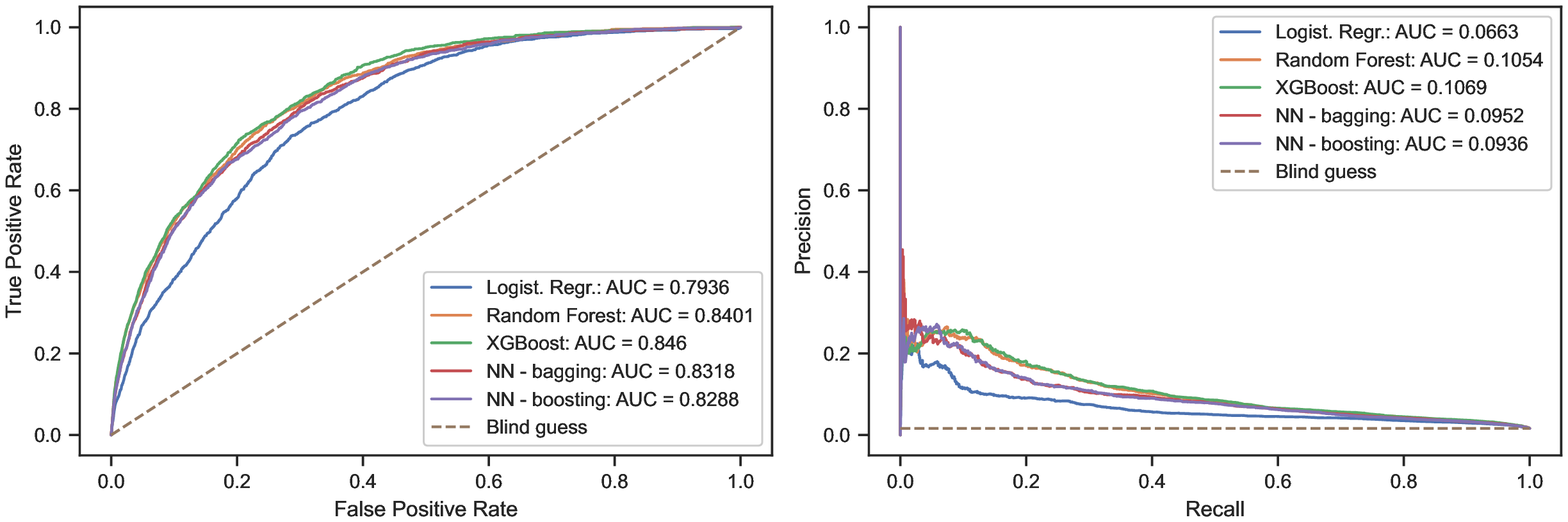}
            \subcaption{Profile 3}
            \label{fig:roc_3}
        \end{subfigure}
        \caption{ROC curves (left) and precision-recall curves (right) on test data $\mathcal{D}_{\text{test}}$ for the respective surrender profiles.}
    \end{figure}
    \begin{figure}[H]
        \ContinuedFloat
        \centering
        \begin{subfigure}{.95\textwidth} 
            \centering
            \includegraphics[width=\textwidth]{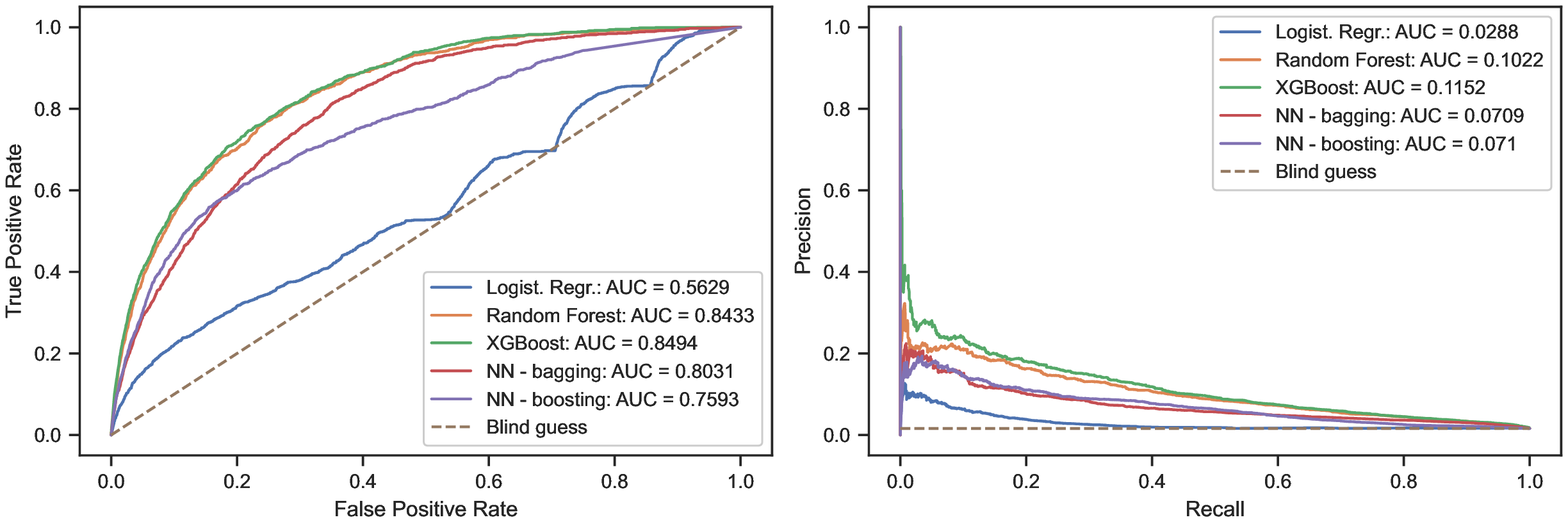}
            \subcaption{Profile 4}
            \label{fig:roc_4}
        \end{subfigure}
        \captionsetup{list=off,format=cont}
        \caption{ROC curves (left) and precision-recall curves (right) on test data $\mathcal{D}_{\text{test}}$ for the respective surrender profiles.}
        \label{fig:my_label}
    \end{figure}
    
    \begin{figure}[H]
        \centering
        \begin{subfigure}{.95\textwidth} 
            \centering
            \includegraphics[width=\textwidth]{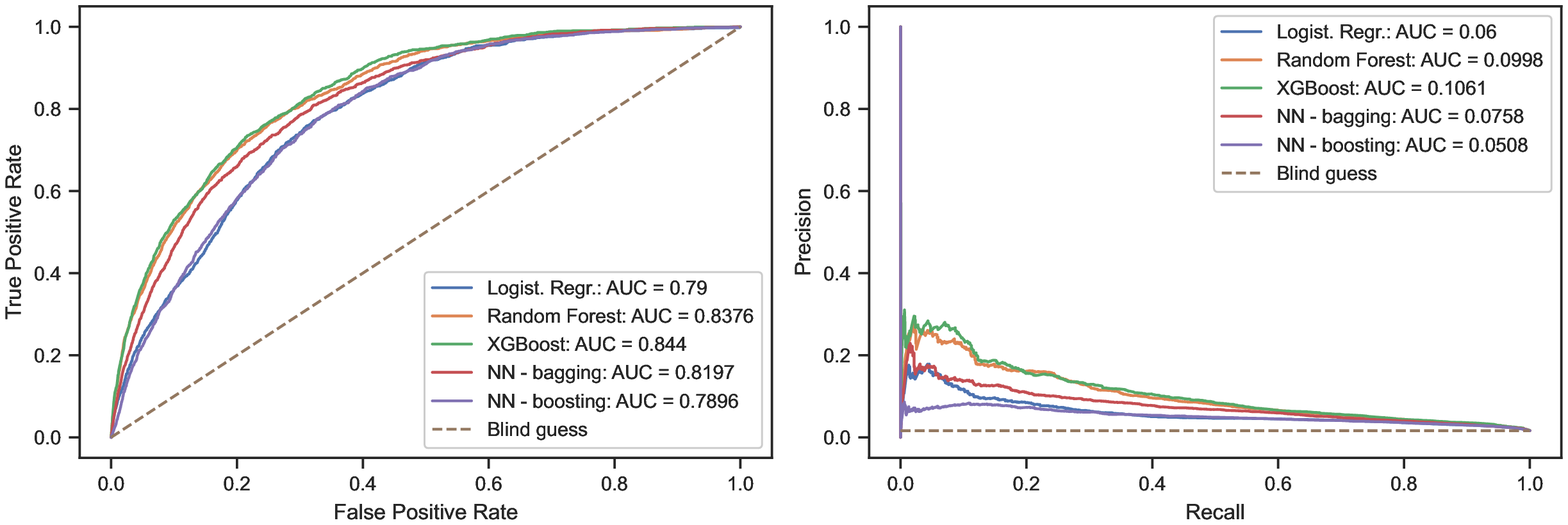}
            \subcaption{Resampling scheme: random undersampling}
            \label{fig:roc_undersampling}
        \end{subfigure}
        \begin{subfigure}{.95\textwidth} 
            \centering
            \includegraphics[width=\textwidth]{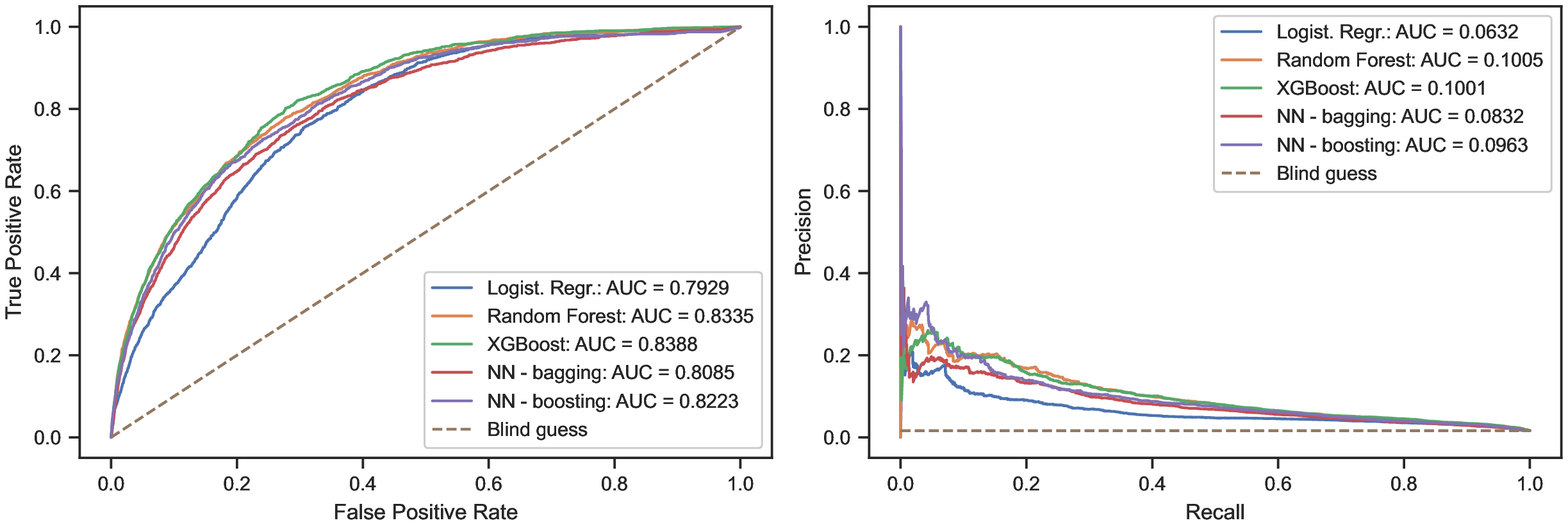}
            \subcaption{Resampling scheme: SMOTE}
            \label{fig:roc_SMOTE}
        \end{subfigure}
        \caption{ROC curves (left) and precision-recall curves (right) on test data $\mathcal{D}_{\text{test}}$ for surrender profile 3. All models were trained on resampled on training data $\mathcal{D}_{\text{train}}$}
        \label{fig:roc_resampling}
    \end{figure}

\newpage
\subsection{Tables} \label{appendix:tables}
\vspace{10pt}
        \begin{table}[htb]
            \centering
            \renewcommand\arraystretch{1.5}
            \begin{tabular}{lcl}
                \hline
                Measure & Formula & Interpretation \\
                \hline
                Accuracy & $\tfrac{TP+TN}{TP+FP+FN+TN}$ & Overall share of correctly classified labels \\
                Precision & $\tfrac{TP}{TP+FP}$ & Share of correct label predictions $\hat{y}=1$ \\
                Recall\footnote{Also known as sensitivity or true positive rate} & $\tfrac{TP}{TP+FN}$ & Share of correctly classified data with $y=1$ \\
                Specificity & $ \tfrac{TN}{FP+TN}$ & Share of correctly classified data with $y=0$ \\  
                False positive rate  & $\tfrac{FP}{FP+TN}$ & Share of misclassified data with $y=0$ \\
                $F_{\beta}$-score & $\tfrac{(1+\beta)\cdot\text{recall}\cdot\text{precision}}{\beta^2\cdot\text{recall}+\text{precision}}$ & Weighted balance ($\beta>0$) between recall and precision \\
                \hline
            \end{tabular}
            \caption{Measures for binary classification, adapted from \cite{Sokolova.2009}. \textcolor{\editcolor}{Abbreviations: True Positive (TP), True Negative (TN), False Positive (FP), False Negative (FN).}}
            \label{table:classification_measures}
        \end{table} 
        
        \begin{table}[htb]
            \centering
            \begin{tabular}{c|cccc|cc}
profile &  $\vert \mathcal{D}_{\text{train}}\vert$ &  imbalance &  RUS &  SMOTE &  $\vert\mathcal{D}_{\text{test}}\vert$ &  imbalance \\
\midrule
1 &      189285 &     0.0290 &     10966 &      367604 &      75414 &    0.0454 \\
2 &      212978 &     0.0165 &      7042 &      418914 &      81113 &    0.0131 \\
3 &      217515 &     0.0201 &      8738 &      426292 &      89671 &    0.0162 \\
4 &      216734 &     0.0210 &      9082 &      424386 &      89105 &    0.0159 \\
\bottomrule
\end{tabular}

            \caption{Overview of number of 1-year observations and imbalance in the training and test set of all surrender profiles. Additionally, we report the number of observations in a perfectly balanced training set resulting from random-undersampling (RUS) and SMOTE-resampling. }
            \label{table:data_review}
        \end{table}
        

\end{document}